\documentclass[oneside]{siamltex}
\usepackage{graphicx}
\usepackage{subfigure}
\usepackage{amsmath}
\usepackage{enumerate}
\usepackage{amssymb}
\usepackage[colorinlistoftodos]{todonotes}
\usepackage{array}
\usepackage[small]{caption}
\usepackage[hypertexnames=false,colorlinks=true,linkcolor=blue,citecolor=blue,
            bookmarksopen=false,bookmarks=false,
            pdfstartview=XYZ,pdffitwindow=true,pdfcenterwindow=true]{hyperref}

\usepackage{todonotes}
\usepackage{booktabs}
\usepackage{siunitx}
\usepackage{braket}
\usepackage{bm} 

\newtheorem{thm}{Theorem}[section]

\newtheorem{lem}[thm]{Lemma}

\numberwithin{equation}{section}

\newcommand{\vect}[1]{\ensuremath{\bm #1}}

\renewcommand{\u}{\vect{u}}

\newcommand{\R}{\mathcal{R}}
\newcommand{\RR}{\mathbb{R}}
\renewcommand{\L}{\mathcal{L}}

\newcommand{\Z}{\mathbb{Z}}

\newcommand {\E} { {\mathbb E} }
\newcommand {\EE} { \mathcal{E} } 

\newcommand{\eps}{\varepsilon}

\newcommand{\B}{\mathcal B}
\newcommand{\BB}{\mathbb{B}}
\newcommand{\amu}{\bar\mu}
\newcommand{\dmu}{\Delta\mu}
\newcommand{\anu}{\bar\nu}
\newcommand{\axi}{\bar\xi}
\newcommand{\azeta}{\bar\zeta}

\DeclareMathOperator{\Id}{{\rm Id}}

\DeclareMathOperator{\var}{{\rm var}}
\DeclareMathOperator{\prob}{{Pr}}
\DeclareMathOperator{\binomial}{Binomial}

\DeclareMathOperator{\erfc}{erfc}

\graphicspath{{Figures/}}

\begin{document}


\title{Noise reduction in coarse bifurcation analysis o1f stochastic agent-based models: an example of consumer lock-in}

\author{
Daniele Avitabile
\thanks{ 
School of Mathematical Sciences, University of Nottingham, Nottingham, NG2 7RD, UK }
\and 
Rebecca Hoyle
\thanks{ 
Department of Mathematics, University of Surrey, Guildford, Surrey GU2 7XH, UK}
\and 
Giovanni Samaey
\thanks{ 
Department of Computer Science, K.U. Leuven, Celestijnenlaan 200A, 3001 Leuven,
Belgium}
}

\maketitle

\begin{abstract}
  We investigate coarse equilibrium states of a fine-scale, stochastic agent-based
  model of consumer lock-in in a duopolistic market. In the model, agents decide on
  their next purchase based on a combination of their personal preference and their
  neighbours' opinions. For agents with independent identically-distributed
  parameters and all-to-all coupling, we derive
  an analytic approximate coarse evolution-map for the expected average purchase.
  We then study the emergence of coarse fronts 
  when the agents are split into two factions with opposite
  preferences.
  We develop a novel Newton-Krylov method that is able to compute accurately and
  efficiently coarse fixed points when the underlying fine-scale dynamics is
  stochastic. The main novelty of the algorithm is in the elimination of the noise
  that is generated when estimating Jacobian-vector products using time-integration of
  perturbed initial conditions.  We present numerical results that demonstrate the
  convergence properties of the numerical method, and use the method to show that
  macroscopic fronts in this model destabilise at a coarse symmetry-breaking
  bifurcation. 
\end{abstract}

\section{Introduction}

Understanding how social groups reach general agreement or perform a coordinated
task has been the subject of an intense research effort over the past fifty years
\cite{Castellano2009a}. In models of social behaviour, consensus is a
macroscopic feature emerging from random reciprocal interactions between a
large number of heterogeneous actors. Understanding how consensus arises and
identifying the key factors for its generation or inhibition are fundamental
questions in social dynamics.

A large class of social models, known as \textit{sociophysical
models}~\cite{Galam1982q}, is based on an analogy
with ferromagnetism: social attributes such as opinions or preferences then correspond to
magnetic dipole moments of atomic spins and choices are
influenced by interactions with neighbouring spins; consensus is then represented by a
phase transition~\cite{Callen1974r,Cipra1987t,Stanley1987q} that is
studied with the tools of statistical mechanics. Sociophysical
models have been applied in various social contexts to study for
instance segregation (Ising and Schelling
models)~\cite{Schelling2006a,Gauvin2009a,Stauffer2007a}, opinion
formation~\cite{Holley1975n,Sznajd-Weron2005n} and social
impact~\cite{Lewenstein1992a}. We refer the reader to the reviews by
Schweitzer~\cite{Schweitzer2002a}, Castellano~\cite{Castellano2009a} and
Chakrabarti and co-editors~\cite{Chakrabarti2007a} for further examples.

In statistical mechanical models, macroscopic coherent states emerge from the
interaction of a large number of identical particles whose behaviour
obeys well-known physical principles. However, particle-like descriptions
of social actors may be seen as simplistic, as, in social systems, individuals do not
behave according to precise physical laws: collective behaviour is the result of the
interaction between complex heterogeneous entities which often take unpredictable
decisions. An alternative strategy is to use agent based models (ABMs)
\cite{Epstein1996y,Axelrod1997a,Gilbert2008a}. ABMs provide a
bottom-up approach to social modeling, in that they focus directly on
individual actors. In ABMs, modelers prescribe detailed rules for agents'
behaviour, possibly including heterogeneities, stochasticity, memory effects and bounded
rationality. Agents exchange information with each other and influence (and are
influenced by) their environment, which may be a model of a physical space or
a network. Because of these characteristics, ABMs have become a popular tool in
social sciences, with applications including crowd dynamics
\cite{Helbing2000a,Moussaid2010a}, civil violence~\cite{Epstein2002a},
urban crime~\cite{Short2010a}, opinion dynamics
\cite{Epstein2001a,Hegselmann2002k,Lorenz2007j,Lorenz2008r} and social networks
\cite{Marsili2004a}. They have also been used to model biological systems
\cite{Graner1992a,Glazier1993a,Troisi2005a,Barnes2010k}. In addition, several ABM
libraries and software packages are available (see \cite{Nikolai2009j} for a review). 

Even though ABMs allow a great level of granularity, it is often interesting to
extract macroscopic variables from the system, study their asymptotic behaviour and
explore their dependence upon control parameters. \textit{Sociodynamical
models}, pioneered by Weidlich \cite{Weidlich1971x,Weidlich2000y}, are obtained
by choosing appropriate coarse variables for the system under
consideration and deriving master equations for the time evolution of their
probability distributions; assuming that the distributions are unimodal and
sharply peaked, an approximate closed nonlinear model for the first few distribution
moments is then derived and analysed with tools from dynamical systems theory. For a
detailed review of techniques and applications of sociodynamics, we refer the
reader to a recent book by Helbing \cite{Helbing2011r}. In general, however, the
induced closure approximations may either be insufficiently accurate, or (in more
complicated situations) impossible to perform, resulting in evolution equations
for coarse variables that are hard or impossible to derive. In those cases, parameter
variations are typically explored via brute-force Monte Carlo simulations, which 
give access only to \textit{stable} asymptotic states and may require long transient
simulations~\cite{Sznajd-Weron2000a,During2009a,Aletti2010a,Sznajd-Weron2011a}.

The past decade has seen a growing interest in the development and deployment of
computational methods that aim at accelerating multiple-scale simulations using
on-the-fly numerical closure approximations.
We mention here equation-free~\cite{Kevrekidis2003r,Kevrekidis2009q} and
heterogeneous multiscale methods~\cite{EEng03,E:2007p3747}. Equation-free methods, in
particular, are an effective tool to bridge
between the microscopic descriptions of sociophysical models or ABMs and the
macroscopic viewpoint of sociodynamical models, since they not only allow for
accelerated simulation at the macroscopic level, but also
enable system-level tasks, such as macroscopic bifurcation analysis. In the
equation-free framework \cite{Kevrekidis2003r,Kevrekidis2009q},
one assumes the existence of a closed macroscopic model in terms of a few
macroscopic state variables. However, instead of deriving an approximate
macroscopic model analytically, one constructs a computational superstructure,
wrapped around a microscopic simulation. In this context, a key tool is the
\textit{coarse time-stepper}, which implements a time step of a macroscopic model that is not available in closed form as a three-step procedure: (i) \textit{lifting},
that is, the creation of initial conditions for the microscopic model,
conditioned upon the macroscopic state at a given time $t$; (ii)
\textit{simulation}, using the microscopic model over a time interval
$[t,t+T]$; and (iii) \textit{restriction}, that is, the estimation of
the macroscopic state at $t+T$. 

While equation-free methods have been employed in various
contexts~\cite{Cisternas2004a,Erban2006a,Laing2006a,Erban2007a,Kolpas2007a,Kolpas2008a,
Laing2010a,Tsoumanis2010a,Spiliotis2011a,Corradi2012a,Hoyle2012a,Marschler2013a},
several numerical issues remain, mainly related to the stochastic nature of the
microscopic evolution.
In the present paper, we focus on some of these numerical aspects
while performing a coarse-grained bifurcation analysis of a stochastic ABM for
opinion formation.  In particular, the computation of macroscopic steady states
requires the solution of a nonlinear system of algebraic equations, which is usually
carried out via Newton-Krylov solvers built around the coarse time-stepper. If the
underlying microscopic evolution equation is stochastic, numerical noise can severely
affect Jacobian evaluations, representing a serious obstacle to the convergence of
the nonlinear iterations~\cite{Hoyle2012a}. 

The present paper deals with the numerical computation of macroscopic coherent
structures for a model of \textit{vendor lock-in}. Lock-in is achieved when
customers repeatedly purchase the same product, irrespective of its quality, because
choosing an alternative vendor is inconvenient or impossible. The term was originally
used to explain the emergence of technological standards, with classic examples
being the prevalence of VHS over Betamax videocassette recorders and of QWERTY
over Dvorak layouts for computer
keyboards~\cite{David1985a,Arthur1989a,Arthur1990a,Janssen1999y}.
The starting point of our investigation is an ABM of vendor lock-in for
duopolistic markets~\cite{Garlick2010k}: Garlick and Chli proposed this model
in order to study, via direct numerical simulations, how to break lock-in. We extend
their model so as to include stochastic dynamics and 
heterogeneities in the agents' preferences
and perform a coarse numerical bifurcation analysis of two types
of macroscopic steady states: a global locked-in state, where the entire agent
population polarizes homogeneously, and 
fronts, which arise when two factions of agents have conflicting preferences.

The present paper thus contains two main contributions. First, for the specific
system under study, we explain the birth of the above-described macroscopic states in
terms of coarse symmetry-breaking bifurcations. To the best of our knowledge, steps
in this direction were taken only very recently~\cite{Spiliotis2012a,Borck2012q} and
were confined to globally locked-in states. In the homogeneous case, we
follow~\cite{Barkley2006y} and interpret metastable locked-in states as fixed points
of a coarse evolution map. In the limit of infinitely many globally-coupled agents
with homogeneous product preferences, we derive the coarse evolution map
analytically. In the case of heterogeneous agents we employ stochastic continuation
and show for the first time how fronts destabilise to 
partially locked-in states. 

The second main contribution of the paper is the development of a novel procedure to
obtain coarse Jacobian-vector products with reduced variance, allowing the accurate
evaluation of Jacobian-vector products in the presence of microscopic stochasticity,
thus gaining full control over the linear and the nonlinear iterations of the
Newton-Krylov solver. Even though our implementation of variance-reduced
Jacobian-vector products is specific to the lock-in model, we believe that analogous
strategies can be applied to other ABMs.  Therefore, we provide a detailed account of
the algorithmic steps involved in defining an accurate equation-free Newton-Krylov
method and testing its convergence properties.

The paper is organised as follows: Section~\ref{sec:model} contains the description
of the lock-in model and a preliminary simulation-based study of coarse macroscopic
states; in Section~\ref{sec:homStates} we derive an approximate analytic coarse map
for the case of homogeneous agents; in Section~\ref{sec:eq-free} we describe the
macroscopic time-stepper for the lock-in model and introduce weighted lifting
operators to obtain variance-reduced Jacobian-vector products; in
Section~\ref{sec:numericalProperties} we test numerical properties of the
Newton-GMRES solver; in Section~\ref{sec:bifurcationResults} we present the results
of the coarse bifurcation analysis and we conclude in Section~\ref{sec:conclusions}.

\section{An ABM for consumer lock-in}\label{sec:model}
\subsection{Model description}

In this section, we introduce a generalization of a consumer lock-in ABM
proposed by Garlick and Chli~\cite{Garlick2010k}, which, in our investigation, 
will serve as a prototypical ABM with heterogeneous agents and binary state
variables.

Let us consider a set of $N$ agents on a two-dimensional square lattice spanning
$[-1,1]^2$, in which the agents are placed on evenly spaced points
$(x_i,y_j)_{i,j=1,1}^{I,J}$, with $x_i=-1+i\Delta x$ and $y_j=-1+j\Delta y$, such
that $x_I=y_J=1$ and $IJ = N$. For notational convenience, we use a lexicographic
numbering of the agents, which are identified by a single index $n$ running from $1$
to $N$. The position of the $n$th agent on the lattice is then denoted by $\vect{r}_n
= (x_n,y_n) = (x_i,y_j)$ with $i = n \mod I $ and $n = (j-1)I+i$.
 
At each discrete time step $t$, agents choose simultaneously between two
products, labelled $0$ and $1$, so that the associated state variables $u_n(t)$
are collected in a vector $\u(t) \in \BB^N$, where $\BB=\set{0,1}$.
Agents are coupled via their neighbourhoods $\square_n$, comprising $\vert
\square_n \vert$ other agents, and their choices are determined by two
parameters: the perceived relative quality $q_n$ of both products, and each agent's
tendency to follow its neighbourhood, $\lambda_n$.  If $q_n \approx -1$, then the
$n$th agent has an intrinsic preference for product $0$ over product $1$ (and the
opposite is true if $q_n \approx 1$). On the other hand, a value $\lambda_n \approx
0$ indicates that the $n$th agent disregards the opinion of its neighbours, whereas
$\lambda_n \approx 1$ implies that the agent aligns itself with the majority of the
neighbours. While these parameters remain constant at all times, each agent draws
its values from an approximate normal distribution, whose moments may depend upon the
position
$\vect{r}_n$,
\begin{equation}
  q_n       \sim \mathcal{N} \big(q; \mu(\vect{r}_n),\xi(\vect{r}_n) \big), \qquad
  \lambda_n \sim \mathcal{N} \big(\lambda; \nu(\vect{r}_n),\zeta(\vect{r}_n) \big), \qquad
  n = 1,\ldots,N.
  \label{eq:parDistr}
\end{equation} 
In practice we require some constraints on $q_n$ and $\lambda_n$, namely $q_n \in
[-1,1]$ and $\lambda_n \in [0,1]$: the normal distributions are chosen such that
this occurs with very high probability; otherwise, the values of $q_n$ and $\lambda_n$ are
discarded and a new random value is generated.  We point out that the issue of
negative parameter values could also be avoided by prescribing distributions which
are naturally defined on finite intervals (for instance, the Beta distribution); we
have chosen the distributions as in the original model by Garlick and
Chli~\cite{Garlick2010k}.

Agent diversity is therefore modelled in two ways: $q_n$ and $\lambda_n$ are randomly
generated and the corresponding probability distributions may vary along the lattice.
In the present paper, we will choose
\begin{equation}
\begin{aligned}
  \mu(\vect{r}_n) &:= \mu(x_n) = \amu + \dmu \tanh( \alpha x_n ), \quad 
  & \xi(\vect{r}_n) = \axi, \\
  \nu(\vect{r}_n) & = \anu,\quad
  & \zeta(\vect{r}_n) = \azeta,
\end{aligned}
\label{eq:parMoments}
\end{equation}
for $n = 1,\ldots,N$, $\amu, \alpha \in \RR$ and $\dmu,\axi,\anu,\azeta \in \RR^+$.
Note that with the above choice, we only introduce a one-dimensional 
parametrization of the mean preferences.
More general multi-dimensional parametrizations of preferences and agent's
tendency to follow their neighbourhood are conceivable, but will not be
considered in this paper. 
As we shall see in the following sections, the sigmoid
$\mu(x_n)$ allows us to model the existence of factions with strong preferences for
one product.

At each time step, agents simultaneously inspect their neighbourhoods and compute two
utility functions, associated with products $0$ and $1$, that represent a weighted
average between their intrinsic preference and the choice of their neighbours,
\begin{equation}
  \begin{aligned}
    f^0_n\big(\u(t)\big) & = - (1-\lambda_n)\frac{q_n}{2} + \lambda_n \Bigg[ 1 -
    \frac{1}{\vert \square_n \vert}\sum_{n' \in \square_n} u_{n'}(t) \Bigg], \\
    f^1_n\big(\u(t) \big) & = (1-\lambda_n)\frac{q_n}{2} + \frac{\lambda_n}{\vert \square_n \vert}\sum_{n' \in \square_n } u_{n'}(t).
  \end{aligned}
  \label{eq:util}
\end{equation}
Once the utility functions have been computed, each agent selects a product at
time $t+1$ according to a Bernoulli distribution whose mean depends upon
the difference between the utility functions at time $t$. More precisely, let 
\[
\Delta f_n\big(\u(t)\big) = f^1_n\big(\u(t)\big) - f^0_n\big(\u(t)\big),
\]
then the $n$th agent's choice is determined via the following conditional
distribution
\begin{equation}
\begin{aligned}
  p( u_n(t+1) & = 1 | \vect{u(t)} = \vect{v} ) = 
  \frac{\exp\big[\beta \Delta f_n (\vect{v}(t))\big]}{
  \exp[-\beta \Delta f_n (\vect{v}(t))]+\exp\big[\beta \Delta f_n (\vect{v}(t))\big]
 }, \\
  p( u_n(t+1) & = 0 | \vect{u(t)} = \vect{v} ) = 
  \frac{\exp\big[-\beta \Delta f_n (\vect{v}(t))\big]}{
  \exp[-\beta \Delta f_n (\vect{v}(t))]+\exp\big[\beta \Delta f_n (\vect{v}(t))\big]
 },
\end{aligned}
\label{eq:p_n}
\end{equation}
for $n = 1,\ldots,N$ and $\beta \in \RR$.
The evolution of the system is best understood by inspecting the utility
functions \eqref{eq:util}. The function $f^1_n$, for instance, is formed by two
contributions: the first addend pertains to the perceived quality of product
$1$; the second addend accounts for the neighbourhood's influence, since this
term is proportional to the number of purchases of product $1$ in the
neighbourhood. The relative importance of the two contributions is determined by
the parameter $\lambda_n$. 

The agent-based model described by \eqref{eq:util}--\eqref{eq:p_n}, completed
by initial conditions and explicit expressions for means and standard
deviations in~\eqref{eq:parDistr}, defines an evolution equation that we will
formally denote by
\begin{equation}
  \begin{aligned}
    & \u(t+1) = \vect{\varphi}( \u(t); \vect{\gamma},\vect{\omega}^n ), \quad t \in \Z_+, \quad \u \in \BB^N,
    \quad \vect{\gamma} \in \RR^Q, \\
    & \u(0) = \u_0,
  \end{aligned}
  \label{eq:phiMap}
\end{equation}
where we have collected in $\vect{\gamma}$ the following microscopic parameters
\begin{equation}
  \vect{\gamma} = (q_1,\ldots,q_n,\lambda_1,\ldots,\lambda_n, \beta),
  \label{eq:microParList}
\end{equation}
and $\vect{\omega}^n$ denotes the set of random choices that were made by the agents
during this time step.

Henceforth, we will refer to~\eqref{eq:phiMap} as the \textit{lock-in model},
implying that the agents behave as specified in \eqref{eq:parDistr}--\eqref{eq:p_n}.
Unless otherwise stated, we shall assume all-to-all coupling, that is, $\square_n =
\set{1,\ldots,N}$ for all $n$, and random Bernoulli-distributed initial conditions
with average $0.5$ 
\begin{equation}
  u_{0n} \sim \mathcal{B}(u;0.5) \qquad \text{$u_{0n}$ i.i.d. for $n=1,\ldots,N$.}
  \label{eq:iniCond}
\end{equation}
\begin{rem}[Interpretation of the coordinates $\vect{r}_n$]
  \label{rem:space}
  Since we have chosen all-to-all coupling and agents are identically coupled via their
  mean preference (see Equation~\eqref{eq:util}), it would be misleading to interpret
  $\vect{r}_n=(x_n,y_n)$ as a location in physical space: with this type of
  coupling the spatial position of the agents does not play any role in the
  evolution of the system. However, from
  Equations~\eqref{eq:parDistr}--\eqref{eq:parMoments} we see that $x_n$ is used
  to order agents by their mean preference, via the sigmoidal function
  $\mu(x_n)$. Such ordering is of course arbitrary, but it allows us to make a
  distinction between two cases: homogeneous agents, when the distribution of
  the quality parameter $q_n$ is the same for all agents ($\Delta \mu = 0$), and
  heterogeneous agents, when the average perceived quality varies within the
  population. In the reminder of the paper, agents will be presented on a
  two-dimensional lattice only for visualisation purposes and the reader should
  interpret $x_n$ as a position in preference space, not physical space. Such
  preference space is parametrised by a the single coordinate $x_n$, as $y_n$
  does not play a role in our simulations.
\end{rem}
\begin{rem}[Deterministic lock-in model]
  The evolution of the lock-in model is stochastic, since
  agents' choices are determined via~\eqref{eq:p_n}.
  However, it is possible to study a deterministic evolution by considering the
  limit $\beta \to \infty$. In this case, agents purchase their
  product according to
  \begin{equation}
    u_n(t+1) = 
    \begin{cases}
      u_n(t) & \text{\textnormal{if $ f^0_n\big(\u(t) \big) = f^1_n\big(\u(t) \big) $,}} \\
      0 & \text{\textnormal{if $ f^0_n\big(\u(t) \big) > f^1_n\big(\u(t) \big) $,}} \\
      1 & \text{\textnormal{if $ f^0_n\big(\u(t) \big) < f^1_n\big(\u(t) \big) $,}}
    \end{cases}\qquad
    n = 1,\ldots,N.
    \label{eq:un_det}
  \end{equation}
  In passing we note that, in the limit $\beta \to \infty$, $\Delta f_n \to 0$,
  Equations~\eqref{eq:p_n}, give a random choice and assign equal probability to $0$
  and $1$; to make the model deterministic, we then prescribe that if
  $f_n^0(\vect{u}(t))=f_n^1(\vect{u}(t))$ the agent sticks with its previous
  decision.
Then, $u_n(t+1)$ is a deterministic function of $\vect{u}(t)$ and of the
  agent's parameters $\lambda_n$, $q_n$. Even in that case, $\vect{u}(t)$ remains
  a random variable, since $\lambda_n$, $q_n$ and the initial condition
  $\vect{u}_0$ are randomly distributed according to
  \eqref{eq:parDistr} and \eqref{eq:iniCond}, respectively. Our model additionally
  differs from the original model of Garlick and Chli~\cite{Garlick2010k} in two
  ways. First, the model
  in~\cite{Garlick2010k} is a deterministic lock-in model with
  no 
  heterogeneity in the agent's preferences,
  $\dmu = 0$. Second, we rescaled the
  utility function so as to include a single
  parameter $q_n$ for the perceived quality, as opposed to having separate
  parameters for products $1$ and $0$.
  \label{rem:deterministic}
\end{rem}
\begin{rem}[Possible model extensions]
  Different types of coupling can be considered for the agents. Beside the
  all-to-all coupling adopted in the present paper,
  nearest-neighbour~\cite{Garlick2010k} and static/dynamic small-world
  couplings~\cite{Watts1998k} are also possible. 
  Considering nearest-neighbour coupling or agent motility would
  effectively introduce a genuine spatial dependence in the system (see
  Remark~\ref{rem:space}).
  More realistic models can also be obtained if the agents adapt their
  parameters $q_n$ and $\lambda_n$ as time varies, so that they can change their
  opinion about the products or their attitude towards the neighbourhood.
\end{rem}

\subsection{Simulation-based study of the lock-in system} 
\label{sec:simulations}
\begin{table}
  \centering
  \begin{tabular}{c c c S[table-format=1.1]  S[table-format=1.4]  S[table-format=1.2]  S[table-format=1.4]  c }
    \toprule
    Experiment & {$\amu$} & {$\dmu$} & {$\alpha$} & {$\axi$} & {$\anu$} & {$\azeta$} & {$\beta$}   \\
    \midrule
             E1 &   0   &  0  &  0    & 0.236  &  0.05 &  0.0167 & 10 \\
             E2 &   0   &  0  &  0    & 0.236  &  0.5  &  0.167  & 10 \\
             E3 &   0   &  1  &  5.0  & 0.236  &  0.5  &  0.167  & 10 \\
             E4 &   0   &  1  &  0.5  & 0.236  &  0.5  &  0.167  & 10 \\
    \bottomrule
  \end{tabular}
  \caption{Parameter values for the lock-in model simulations of
  Figures~\ref{fig:homStates} and~\ref{fig:inhomStates}. Experiments are done
  with all-to-all coupling and random Bernoulli-distributed initial
  conditions~\eqref{eq:iniCond}.}
  \label{tab:params}
\end{table}
\begin{figure}
  \centering
  \includegraphics{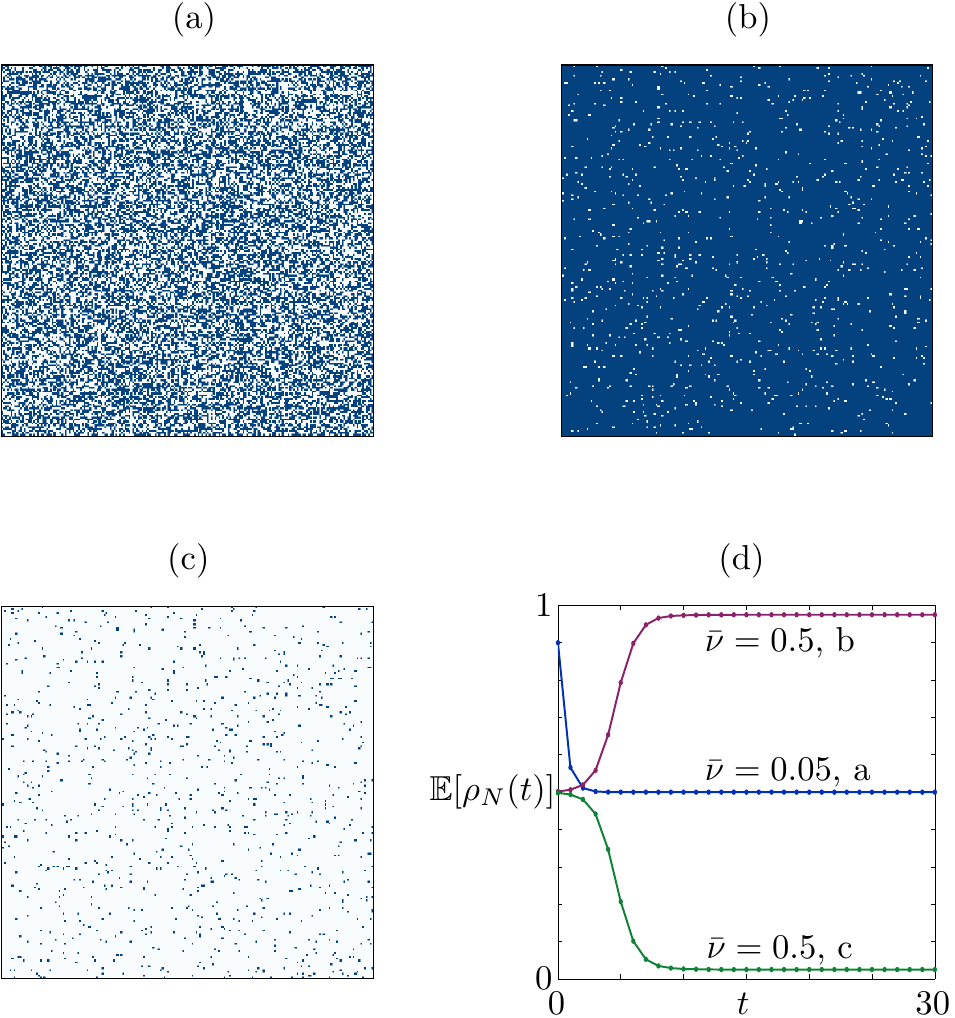}
  \caption{Homogeneous states of the lock-in model
  \eqref{eq:parDistr}--\eqref{eq:p_n} with $N = 200^2$ agents, corresponding to
  experiments E1 and E2 in Table~\ref{tab:params}. White and blue
  dots represent agents purchasing product $0$ and $1$ respectively. The lock-in
  model is initialised sampling a Bernoulli distribution~\eqref{eq:iniCond} and
  it is iterated for $30$ time steps. Panel (a): mixed state obtained with
  experiment E1. Panel (b): global lock-in of product $1$, obtained with one
  realization of E2. Panel (c): global lock-in of product $0$, obtained with a
  second realization of E2. Panel (d): ensemble average of the average purchase
  $\rho_N = \sum_n u_n/N $ over $2000$ realizations
  as a function of time for simulations of panels (a), (b) and (c); the
  experiment of panel (a) is here repeated with initial conditions $u_{0n} \sim
  \B(u;0.9)$ for $n=0,\ldots,N$, showing that the mixed state is the unique
  macroscopic stable equilibrium for $\anu = 0.05$.}
  \label{fig:homStates}
\end{figure}

We now discuss microscopic numerical simulations that motivate our choice of the
macroscopic state variables. In the following numerical experiments, we iterate the
lock-in model~\eqref{eq:phiMap} with initial condition~\eqref{eq:iniCond} for the
choices of the parameter distributions~\eqref{eq:parDistr} specified in
Table~\ref{tab:params}. This leads naturally to the introduction of a set of
macroscopic variables, which will be defined more precisely in
Section~\ref{subsec:macroscopicVariables}.

\subsubsection{Globally locked-in states with homogeneous agents}

With the first two experiments, using parameter sets E1 and E2, we find
homogeneous
macroscopic solutions corresponding to globally locked-in states. 
In Figures~\ref{fig:homStates}(a)--\ref{fig:homStates}(c) we show 1
\textit{mixed state} obtained in E1 and 2 \textit{locked-in states} obtained in
E2. In each realization of these experiments, we obtain different final states
since the evolution is stochastic and the initial condition $\u_0$ as well as the
microscopic parameters $\vect{\gamma}$ are randomly distributed. 
We recall that agents are presented on a two-dimensional lattice
for visualisation purposes, but their position $\vect{r}_n$ on the lattice does not
influence the dynamics (see Remark~\ref{rem:space}).

In experiment E1, we set the agent's parameters so that the average perceived
quality of product 0 and 1 is identical and the tendency to follow the
neighbourhood is low (see Table~\ref{tab:params} and
Figure~\ref{fig:homStates}(a)). The resulting state is a \textit{mixed
state}, with an even distribution of final products. In experiment E2, we
increase the average and variance of the coupling
(Figures~\ref{fig:homStates}(b)--\ref{fig:homStates}(c))
and observe two \textit{locked-in states} (each equally likely to occur) in which
almost all agents continually purchase one product, irrespective of its
perceived quality. Indeed, since 
$\amu = \Delta \mu = 0$, 
we expect that on average
only half of the agents have a preference for the dominant product, whereas agents in
the remaining half purchase a product that they consider worse in terms of quality.
As the experiment is repeated, we can get lock-in of either product, owing to
the stochasticity of the evolution and the randomness of microscopic parameters and initial
conditions. These results are in accordance with what was reported by
Garlick and Chli~\cite{Garlick2010k} for a deterministic lock-in model with
all-to-all coupling (see also Remark~\ref{rem:deterministic}) and reinforce the
similarity between the lock-in ABM and other Ising-type sociophysical models
available in the literature~\cite{Castellano2009a}.

It is natural to seek for a characterization of the lock-in model in terms of a
simple macroscopic variable and to interpret the statistical
equilibria obtained as steady states of a suitably-defined dynamical system. In panel
Figure~\ref{fig:homStates}(d) we begin introducing such a characterization:
we repeat $2000$ times the numerical simulations that led to each of the states
in panels (a), (b), (c), group each of the samples by their 
mean purchase
\[
\rho_N(t) = \frac{1}{N} \sum_n u_n(t) \in \mathbb{Q}_N, 
\quad t \in \mathbb{Z}_+,
\quad \textrm{where } \mathbb{Q}_N = \Set{
\frac{n}{N} \in
\mathbb{Q} | 0 \leq n \leq N}, 
\]
 and plot the ensemble average of these 
 means
as a function of time. 
The macroscopic variable $\rho_N$ is a scalar, as agent's
preferences do not depend on $\vect{r}_n$. The plot in
Figure~\ref{fig:homStates}(d) shows that, in this description,
locked-in and mixed states are achieved rapidly, within just $10$ iterations of
the map. For a low value of the
average coupling strength $\bar \nu$, the system reaches a single macroscopic state: from
panel (d) we see that, in this region of parameter space, the
mixed state is attracting even if the initial conditions are close to a fully
locked-in state, that is, $u_{0n} \sim \B(u;0.9)$, for $n=1,\ldots,N$. Upon
increasing the coupling strength, we find two new macroscopic states, suggesting the
presence of a pitchfork bifurcation at the macroscopic level.  
\begin{figure}
  \centering
  \includegraphics{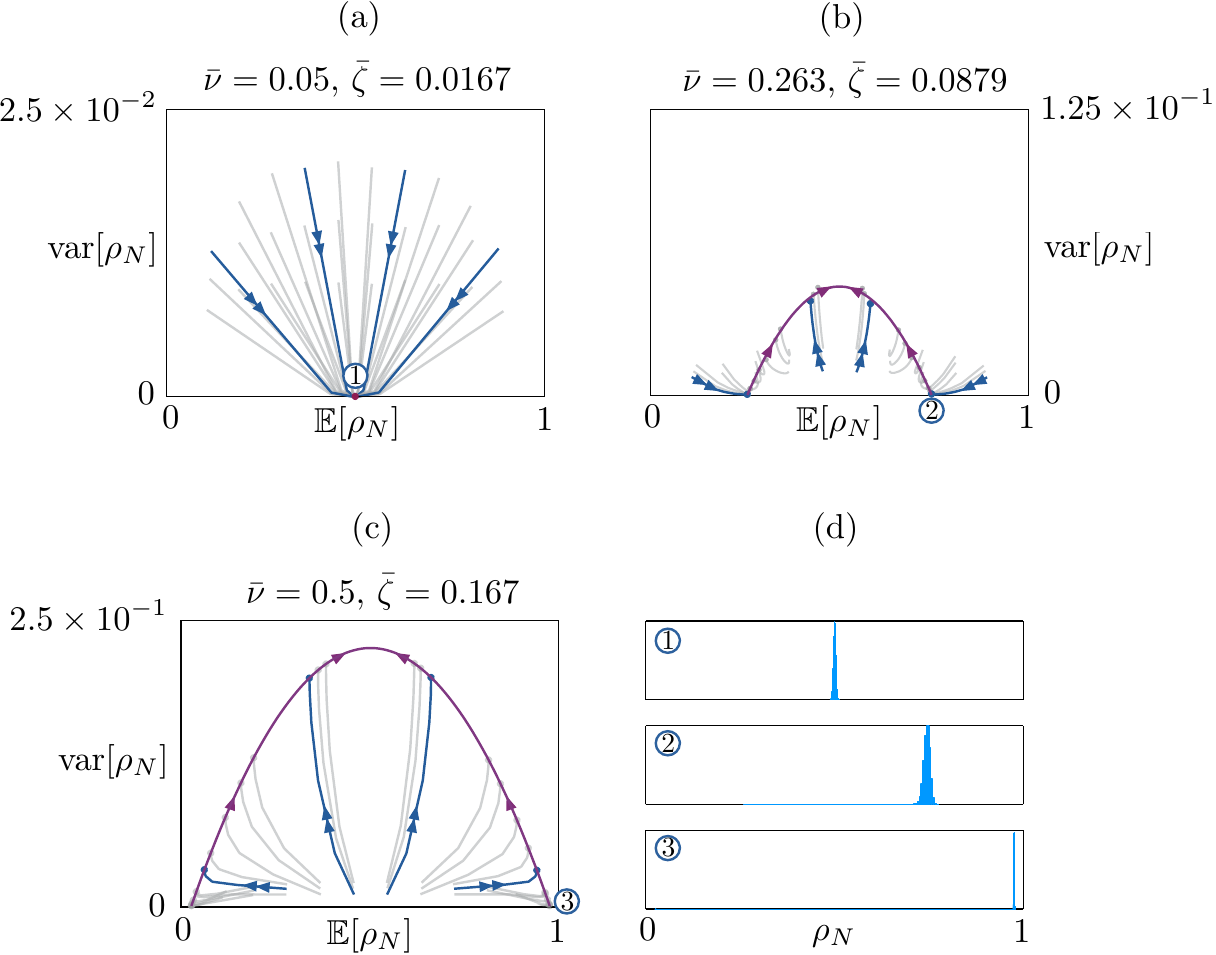}
  \caption{Expectation and variance of the average agents' choice $\rho_N(t) =
  \sum_n u_n(t)$ starting from different initial conditions. For each trajectory on
  the $(\mathbb{E}[\rho_N],\textrm{var}[\rho_N])$-plane, we initialise $M=10^5$
  independent simulations for $N=200^2$ agents (as in Figure~\ref{fig:homStates})
  with various initial conditions and iterate the lock-in model until $t=20$. Panel
  (a): low average coupling $\lambda_n$ (as in E1 of Table~\ref{tab:params}). Panel
  (b): intermediate value of the average coupling. Panel (c): high value of the
  average coupling (as in E2). Panel(d): examples of final distributions of
  $\rho_N$.}
  \label{fig:slaving}
\end{figure}
\begin{figure}
  \centering
  \includegraphics{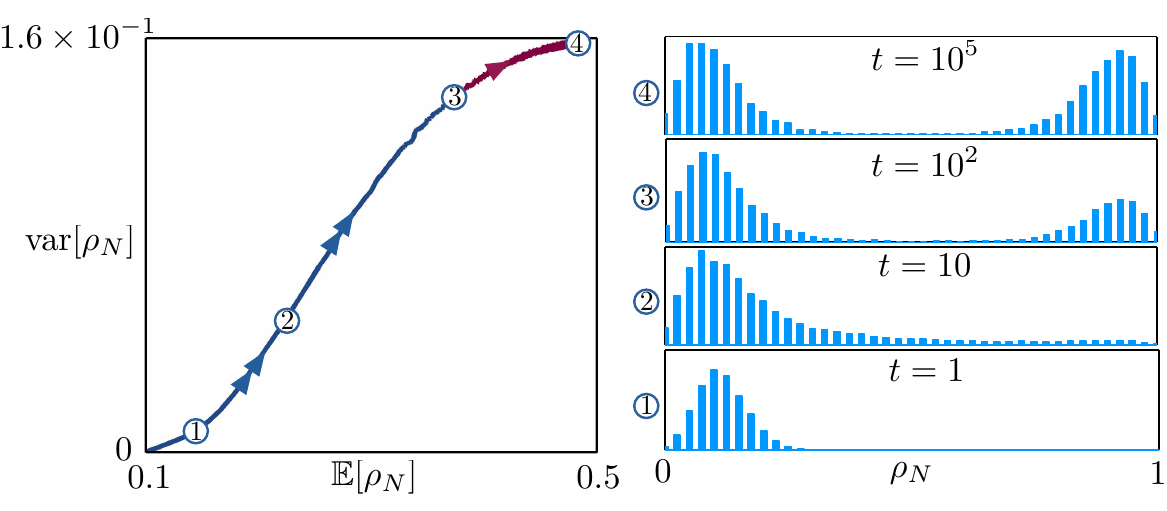}
  \caption{Expectation and variance of the average agents' choice $\rho_N(t) =
  \sum_n u_n(t)/N$ for the lock-in model with $N=40$ agents (see also the animation
  \href{run:Movies/slaving.avi}{slaving.avi}). Left: we
  initialise $M=10^4$ realizations from $u_{0n}=0.1$ for all $n$; the evolution on
  the $(\mathbb{E}[\rho_N],\textrm{var}[\rho_N])$-plane is plotted with a blue line
  for the first $10^2$ iterations and in magenta for the following iterations, until
  $t=10^5$. Right: distributions of $\rho_N$ at various times. After a short
  transient, the evolution takes place on a slow manifold. 
  Parameters: $\amu = 0$, $\Delta\mu = 0$, $\alpha=0$, $\axi=0.235$, $\anu=0.3$,
  $\azeta=0.03$, $\beta=8$.}
  \label{fig:slowManifold}
\end{figure}

However, a more careful inspection shows that these macroscopic locked-in solutions
are not stable steady states, but rather coarse metastable states:  it is indeed known
that, in sociophysical models, the lifetime of metastable
states is linked to finite system size \cite{Castellano2009a}.
In Figure~\ref{fig:slaving}, we repeat similar computations and monitor $\E[\rho_N]$
and $\var[\rho_N]$ as a function of time. This time we prepare realizations with different
initial expectation and variance, so as to plot several
orbits on the $(\E[\rho_N],\var[\rho_N])$-plane. A low average value of the
coupling parameter $\lambda_n$ leads to a single steady state, as shown in
Figure~\ref{fig:slaving}(a). For intermediate and high values of the
coupling (Figures~\ref{fig:slaving}(b) and \ref{fig:slaving}(c)), trajectories are
quickly attracted to a slow manifold (purple curve) which, for
these choices of parameters, is well approximated by a parabola. 
In equation-free terminology, the existence of a slow manifold in the
$(\E[\rho_N],\var[\rho_N])$-plane is referred to as \textit{slaving}. 

Asymptotic equilibria in Figure~\ref{fig:slaving}(a) have small variance (they correspond to
sharply peaked distributions with average equal to $0.5$) while asymptotic equilibria in
Figures~\ref{fig:slaving}(b) and \ref{fig:slaving}(c) have a much higher variance.
The latter distributions have means equal to $0.5$, but they are bimodal (as
will be shown below).

In Figure~\ref{fig:slaving}, we iterate the lock-in model only until $t=20$, a
time scale clearly suggested by the coarse solution curves of
Figure~\ref{fig:homStates}(d): in fact, for these choices of the control
parameters, the time scale of the drifting on the
slow manifold is so long that it is not feasible to observe it with numerical
computations; hence the magenta curves in
Figures~\ref{fig:slaving}(b)--\ref{fig:slaving}(c) are
obtained by fitting a parabola to the set of final points on the phase plane. In
Figure~\ref{fig:slowManifold} (and the accompanying animation
\href{run:Movies/slaving.avi}{slaving.avi}), the system size and parameters have
been adjusted to observe drifting on more affordable time scales ($N=40$, $\amu
= 0$, $\Delta\mu = 0$, $\alpha=0$, $\axi=0.235$, $\anu=0.3$, $\azeta=0.03$ and
$\beta=8$): the initial probability distribution of $\rho_N$ is a Dirac delta,
which becomes a unimodal distribution with nonzero variance on time scales of
order $t=10$ and drifts towards a bimodal distribution on time scales of order
$t=10^5$. The system therefore always evolves towards a state with
$\E[\rho_N]=1/2$. However, the difference between strong coupling and weak
coupling is clearly visible: when the agents are weakly coupled, each individual
realization of the system evolves to a mixed state with $\rho_N = 1/2$, so there
is no lock-in, whereas with strong coupling between the agents, each realization
will display lock-in and the initial condition determines which state the agents
will be locked into. Due to microscopic stochasticity, the system only
equilibrates over a very long time scale, over which a fraction of the
realizations flips to the other locked-in state in the latter case. 

As a consequence, even though these locked-in solutions are only metastable, it is still
meaningful to characterize them as fixed points of an evolution map  
on intermediate time-scales. Barkley, Kevrekidis and Stuart studied metastable
states in physical systems with similar properties and use the term \emph{moment
map} for the coarse evolution operator~\cite{Barkley2006y}.
We shall return to this moment map for homogeneous steady states in
Section~\ref{sec:homStates}, where we derive an approximate coarse evolution map
for the lock-in model. 

\subsubsection{Fronts for heterogeneous agents} We now turn to
heterogeneous states, which correspond to large-dimensional coarse maps
and are more challenging to compute with equation-free methods.
\begin{figure}
  \centering
  \includegraphics{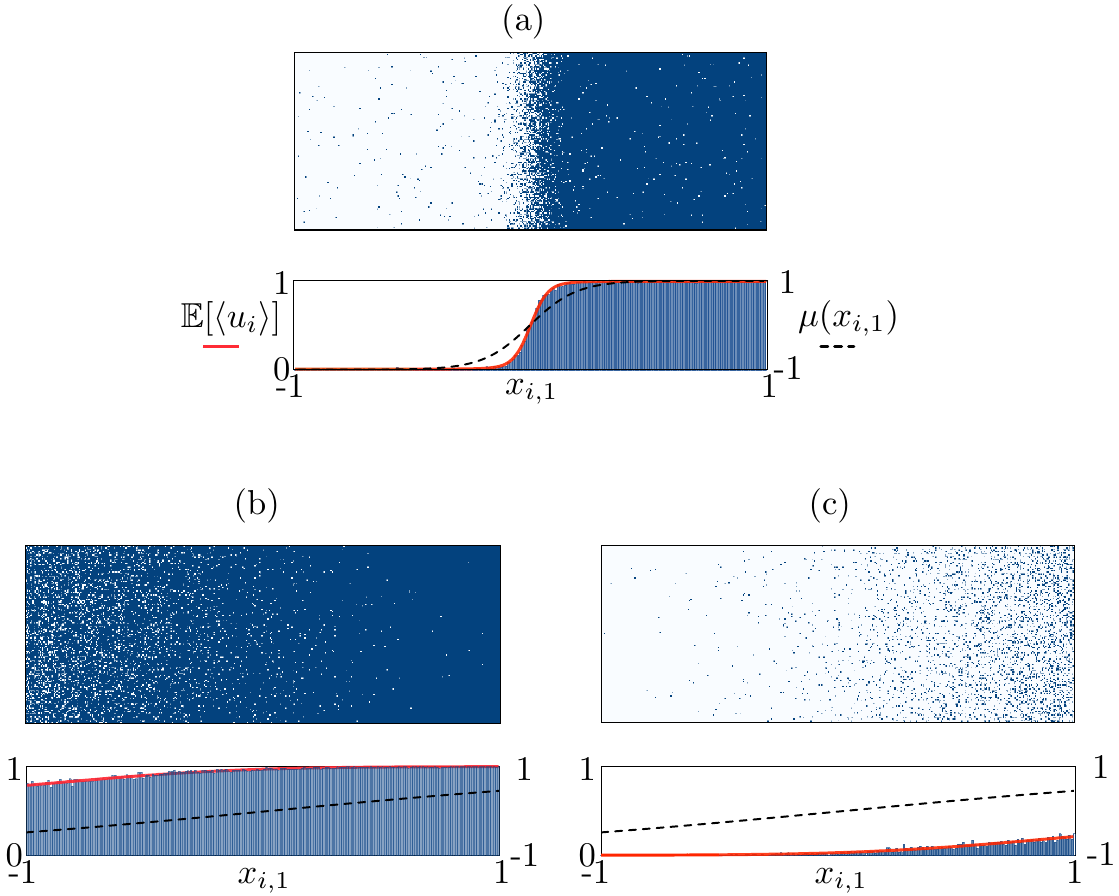}
  \caption{Heterogeneous states of the lock-in model
  \eqref{eq:parDistr}--\eqref{eq:p_n} with $N = 400 \times 150$ agents,
  corresponding to experiments E3 and E4 in Table~\ref{tab:params} (see also the
  animations 
  \href{run:Movies/stableFront.avi}{stableFront.avi}
  and
  \href{run:Movies/unstableFront.avi}{unstableFront.avi}
  ). The square $[-1,1]^2$ has been scaled to a rectangle for visualization purposes
  and we use $x_{i,j}$ to indicate the position on the lattice. The lock-in
  model is initialised sampling a Bernoulli distribution~\eqref{eq:iniCond} and
  it is iterated for $30$ time steps. Panel (a): in experiment E3 a stable
  interface is formed between two locked-in states (top); macroscopic states are
  obtained by averaging along the $y$ axis (bottom, blue histogram) and then
  taking an ensemble average with respect to $2000$ realizations (bottom, red
  curve, left axis);
  The resulting macroscopic state is a front connecting a macroscopic
  $0$-state with a macroscopic $1$-state; the front is sharper than the 
  profile $\mu$ of the average quality perception parameters $q_n$
  (bottom, dashed black curve, right axis). Panels (b) and
  (c): in experiment E4, the slope of the profile $\mu$ is varied; the
  macroscopic front loses stability, giving rise to two stable heterogeneous
  states featuring pockets of resistance.}
  \label{fig:inhomStates}
\end{figure}
More specifically, we allow the average quality perception
$\mu$ to vary within the population, as specified
in~\eqref{eq:parMoments}: in experiments E3 and E4 we choose a nonzero $\dmu$, and
vary the steepness $\alpha$ of the sigmoidal function $\mu(x_n)$. 
In E3, for instance, we choose $\amu = 0$, $\dmu = 1$ and $\alpha = 5$, so the agents
are split into two factions of the same size: the coordinate $x$ parametrises
the agents' mean preferences, so agents with negative $x$ like product $0$ and
agents with positive $x$ favour product $1$. In E4, the sigmoid $\mu(x)$ is less
steep, meaning that the two factions have still the same size, but there are fewer
zealots.


For these experiments we choose a lattice of $400 \times 150$ agents, initialise the
system with the Bernoulli distribution with success probability $0.5$ and evolve the
map for $30$ iterations. 
As in the homogeneous case, the two-dimensional lattice is used only for
visualization purposes and the position on the lattice should not be interpreted as a
physical location, but as a position in preference space (see
Remark~\ref{rem:space}).


Figure~\ref{fig:inhomStates}(a) shows the results of experiment E3, for which
$\alpha = 5$.
%
%
As expected, the inhomogeneity in the distribution of the average
quality perception parameters $q_n$ induces the formation of 
a pattern. In the accompanying animation
\href{run:Movies/stableFront.avi}{stableFront.avi} we initialise the
system using a slightly different initial condition, $u_{0n} = H(x_n)$, where $H$ is
the Heaviside function. The state represented in Figure~\ref{fig:inhomStates}(a)
is the only attracting solution in this region of parameter space.
 
We shall abandon for a moment the lexicographical
ordering used so far and denote the agents' purchases as $u_{i,j}$. To obtain a
macroscopic description of the state in Figure~\ref{fig:inhomStates}(a), we compute
averages of the purchases along the $y$-axis, $\braket{u_i}$, and then
take an ensemble average with respect to $2000$ realizations of the same
experiment (red curve in panel a). The resulting macroscopic state is a front
connecting a macroscopic $0$-state to a $1$-state: the front is
parametrised by the agent's mean preference or, equivalently, by $x$.
For convenience, we also plot $\mu(x_{i,1})$ and compare it to the macroscopic
front, noting that the final
macroscopic steady state is sharper than the  profile of $\mu$.
As we decrease $\alpha$, the macroscopic front persists and becomes flatter,
until a critical point at which two new inhomogeneous states emerge (see the
animation \href{run:Movies/unstableFront.avi}{unstableFront.avi}). Such
states, obtained with experiment E4 and shown in Figure~\ref{fig:inhomStates}(b)
and \ref{fig:inhomStates}(c), are related via the transformation 
\[
\E[\braket{u_i}] \mapsto -\E[\braket{u_{-i}}] + 1.
\]
The scenario described above suggests that, as $\alpha$ is decreased, the front
of Figure~\ref{fig:inhomStates}(a) undergoes a symmetry-breaking bifurcation at the
macroscopic level. To the best of our knowledge, this type of transition has not
been observed before in studies of opinion formation models. In the following
sections, we will give a more precise definition of the macroscopic
variables chosen to describe the lock-in systems for both homogeneous and
inhomogeneous states, and then proceed to perform a numerical bifurcation analysis of the
corresponding states.

\subsection{Macroscopic level description}
\label{subsec:macroscopicVariables}
Let us consider the lock-in model for fixed values of the microscopic parameters $\vect{\gamma}
\in \RR^Q$, which are randomly distributed via~\eqref{eq:parDistr} and remain
constant at all times\footnote{The vector $\vect{\gamma}$, as given
by~\eqref{eq:microParList}, also contains the deterministic parameter $\beta$,
which has been omitted here for simplicity.}.
Then, we denote by $\chi(\vect{u},t|\vect{\gamma})$ the probability distribution
of the vector $\vect{u}(t)$, given these microscopic parameters $\vect{\gamma}$. Now,
considering that the microscopic parameters themselves are distributed according to a
probability distribution $\psi(\vect{\gamma};\vect{\Gamma})$ that depends on a
(small) number of macroscopic parameters
$\vect{\Gamma}=(\bar{\mu},\Delta\mu,\alpha,\bar{\xi},\bar{\nu},\bar{\zeta})\in\mathbb{R}^P$,
we can define the joint
probability distribution of microscopic parameters and states as 
\[ 
p(\vect{u},\vect{\gamma},t) = \chi(\vect{u},t|\vect{\gamma}) \psi(\vect{\gamma};\vect{\Gamma}).
\]
We note that, to simplify notation, we often omit the explicit
dependence of $\psi$ on $\vect{\Gamma}$.
The probability distribution $P(\vect{u},t)$ for an average agent at time $t$ is then obtained by integrating over all possible microscopic parameter values,
\[
P(\vect{u},t) = \int_{\RR^Q}p(\vect{u},\vect{\gamma},t) \, d\vect{\gamma}.
\]
We can formally write the time evolution of $P(\vect{u},t)$ as 
\[
P(\u,t+1)= \int_{\RR^Q} \int_{\BB^{N}} \Psi(\vect{u} | \vect{v},\vect{\gamma}) 
p(\vect{v},t,\vect{\gamma})\, d\vect{v} \, d\vect{\gamma},
\]
where $\Psi(\vect{u} | \vect{v},\vect{\gamma})$ represents the \textit{transition
kernel}, that is, the probability distribution of the state at time $t+1$ given
that the system was in $\vect{v}$ at time $t$ with constant microscopic parameters
$\vect{\gamma}$. \footnote{The transition kernel depends explicitly on time,
$\Psi(\vect{u},t+1|\vect{v},t,\vect{\gamma})$. However, time dependence has been
omitted here to simplify the notation.} 

The macroscopic state $\vect{U}(t)=\left(U_n(t)\right)_{n=1}^N$ that was
described informally in Section~\ref{sec:simulations} is the ensemble average of
a large number of realizations, each with different microscopic parameters. In the limit of
infinitely many realizations ($M\to\infty$), this corresponds to
taking the expectation of $\vect{u}$ with respect to the probability
distribution $P(\vect{u},t)$ 
of the microscopic realizations, 
\begin{equation}
  \vect{U}(t):=\E\left[\vect{u}(t)\right]=\int_{\BB^{N}}
  \u \, P(\vect{u},t) \,d\vect{u}, 
  \label{eq:macro_var} 
\end{equation} 
leading to the evolution map
\begin{equation}
    \vect{U}(t+1)=
    \int_{\BB^{N}} \vect{u} 
    \bigg[
    \int_{\RR^Q}\int_{\BB^{N}} 
    \Psi(\vect{u}|\vect{v},\vect{\gamma}) 
    p(\vect{v},\vect{\gamma},t) \, d\vect{v} \, d\vect{\gamma}\bigg] \, d\vect{u}.
\label{eq:macro_evolution} 
\end{equation} 
Clearly, the macroscopic evolution above cannot be written as a closed form
equation that depends explicitly on $\vect{U}(t)$, unless one makes a closure
approximation that specifies $p(\vect{u},\vect{\gamma},t)$ as a function of
$\vect{U}(t)$.
The focus of the present paper is to obtain bifurcation diagrams for fixed
points of the coarse map~\eqref{eq:macro_evolution}.
The algorithm that will be presented in Section~\ref{sec:eq-free} is a procedure
to impose the aforementioned closure
approximation numerically.

\begin{rem}[Low-dimensional coarse descriptions] Our macroscopic description is
    high-dimensional, in that $\vect{U}$ is a vector with $N$ entries.
    Lower-dimensional descriptions can be obtained expressing $\vect{U}$ in
    terms of a coarse polynomial or spectral
    basis~\cite{Gear2002a,Runborg2002f,Laing2006a}. Since we aim to develop a
    numerical framework suitable for high-dimensional coarse systems (and
    applicable to the low-dimensional descriptions as well), we
    will continue to use simple agent-wise coarse variables in this paper.
\end{rem}
\begin{rem}[Discrete distributions]\label{rem:mc_weights}
  Since $\vect{u}\in \BB^N$, the microscopic state belongs to a discrete set of
  possible admissible states with 
  cardinality
  $2^N$. Thus, the probability
  distribution can be written as 
  \[
    P(\vect{u},t)=\sum_{\vect{v}\in\BB^N}P^*(\vect{v},t)\delta(\vect{u}-\vect{v}),
  \]
  and integrals of the type~\eqref{eq:macro_var}
  should be interpreted as a discrete sum,
  \begin{equation}
    \int_{\BB^n}f(\vect{u}) P(\vect{u},t) \, d\vect{u}=
    \sum_{\vect{u}\in \BB^n}f(\vect{u})P^*(\vect{u},t).
    \label{eq:int_bernouilli}
  \end{equation}
  In other words, the integral is computed assigning to each possible
  configuration $\vect{u}$ a weight corresponding to its probability
  $P^*(\vect{v},t)$. In practical  computations, however, we will not be able to
  simulate all possible realizations $\u$, so we will approximate the integrals
  by a Monte Carlo estimate using $M\ll 2^N$ realizations,
  \begin{equation}
    \int_{\BB^n}f(\vect{u}) P(\vect{u},t)d\vect{u}\approx
    \dfrac{1}{M}\sum_{m=1}^Mf(\vect{u}^m),		
    \label{eq:int_mc}
  \end{equation}
  where $\vect{u}^m$ are sampled from the probability distribution $P(\vect{u},t)$.
\end{rem}

In Section~\ref{sec:simulations} we have introduced homogeneous and inhomogeneous
macroscopic states that we are now ready to characterise by means of coarse
bifurcation analysis: for the former, a simple one-dimensional coarse
description exists and will be discussed in the following section; for the
latter, we will use equation-free bifurcation analysis, which will be the
subject of Sections~\ref{sec:eq-free}--\ref{sec:bifurcationResults}.

\section{Homogeneous macroscopic states}
\label{sec:homStates}
We begin by characterising 
homogeneous macroscopic states, which are
described in terms of the average purchase
\begin{equation}\label{eq:rho_N}
\rho_N(t) = \frac{1}{N} \sum_n u_n(t) \in \mathbb{Q}_N.
\end{equation}
For each $t \in \mathbb{Z}^+$, $\rho_N(t)$ is a random variable, whose probability
distribution is denoted by 
\[
\bar{P}(\rho_N,t) = \int_{\Sigma_{\rho_n}}P(\vect{u},t)d\vect{u},
\qquad
\Sigma_{\rho_n} = \Set{ \vect{u} \in \BB^N | \rho_N = \frac{1}{N} \sum_n u_n}.
\] 

The numerical simulations of
Figure~\ref{fig:slaving} lead us to search for a coarse evolution map whose fixed
points correspond to the homogeneous metastable locked-in states of the lock-in
model. Following~\cite{Barkley2006y}, we search for a first moment map, that is,
a map that closes at the level of the first moment $\E[\rho_N(t)]$ of the
probability distribution $\bar{P}(\rho_N,t)$. In this section we show that a first moment map can be found explicitly under suitable hypotheses.
\begin{lem}\label{lem:coarseMap}
  Let us consider the lock-in model~\eqref{eq:phiMap} under the following hypotheses
  \begin{enumerate}
    \item All-to-all coupling, $\square_n = \set{1,\ldots,N}$ for all $n$.
    \item Deterministic evolution, that is, $\beta \to \infty$.
    \item Deterministic tendency to follow the neighbours 
      \[
      \lambda_n \sim \delta( \lambda - \anu), \quad 
      \anu \in (0,1), \quad n=1,\ldots,N.
      \]
    \item Homogeneous distribution of the quality perception $q_n
      \sim \mathcal{N}(q;\amu,\axi)$, $q_n$ \textnormal{i.i.d.}
  \end{enumerate}
  Further, let $\rho_N(t)$ be the 
 mean purchase
  as defined in~\eqref{eq:rho_N}.
  Then, in the limit as $N \to \infty$, we have 
  \[
  \E [ \rho_\infty(t + 1) ]  =  
     \frac{1}{2} \erfc \bigg[ \frac{1}{\axi \sqrt{2}} \bigg(
	 \anu\frac{1-2\E[\rho_\infty(t)]}{1-\anu} -\amu \bigg) \bigg] +
	 \mathcal{O}\big(\var[\rho_\infty(t)]\big).
  \]
\end{lem}
\begin{proof}
  Hypotheses 1 and 2 imply that the state of an individual agent $u_n(t+1)$ is a
  deterministic function of  
  $\rho_N(t)$ and the individual
  perceived quality $q_n$, which are both random quantities (see
  Remark~\ref{rem:deterministic}). Furthermore, Hypothesis 3 implies
  \footnote{ 
  We use Equation~\eqref{eq:un_det} and omit the case $ f^0_n\big(\u(t)
  \big) = f^1_n\big(\u(t) \big) $, which corresponds to an event of measure $0$.
  }
  \begin{equation}
    u_n(t+1 | \rho_N(t), q_n, \anu) =
    \begin{cases}
      1 & \text{if $q_n >\dfrac{\anu}{ 1 - \anu}( 1- 2 \rho_N(t))$,} \\[1em]
      0 & \text{otherwise,}
    \end{cases} \qquad
      n = 1,\ldots,N.
  \label{eq:unDet}
  \end{equation}

  Next, let us denote by $\bar{p}_n(\rho_N,t,q)$ the joint probability of obtaining 
  a mean purchase
  $\rho_N$ and a perceived quality $q$ for agent $n$. Owing to
 Hypotheses 1 and 4, the agent-wise expectation of $u_n$ with respect to all
 possible realizations of the microscopic parameters,
\begin{equation}
  \E[u_n(t+1)] = \int_{\mathbb{Q}_N} \int_\mathbb{R} u_n(t+1| \rho_N, q, \anu) 
    \, \bar{p}_n(\rho_N,t,q) \,  \, dq d\rho_N, 
  \end{equation}
is the same for all $n$, since $\bar{p}_n(\rho_N,t,q)=\bar{p}(\rho_N,t,q)$,
independently of $n$. Hence,
  \[
    \E[\rho_N(t)] = \frac{1}{N} \sum_n \E[ u_n(t) ] = \E[ u_n(t) ].
  \]
  Similarly, we write $\bar{p}(\rho_N,t,q)=\tilde{p}(\rho_N,t | q)
  \mathcal{N}(q;\amu,\axi)$, in which $\mathcal{N}(q;\amu,\axi)$ is the
  probability density of the Gaussian distribution from which $q$ was drawn. In
  the limit as $N$ tends to infinity, we moreover have that
  $\tilde{p}(\rho_N,t|q)=\bar{P}(\rho_N,t)$, as 
  the mean purchase
  is then independent of a specific agent's perceived quality. We then
  use~\eqref{eq:unDet} to obtain
  \begin{equation}
  \begin{split}
  \E [ \rho_N(t + 1) ] & = \E[ u_n(t+1) ] \\
    & \approx \int_{\mathbb{Q}_N} \int_\mathbb{R} u_n(t+1|\rho_N, q, \anu) 
      \, \bar{P}(\rho_N,t) \mathcal{N}(q;\amu,\axi)  \, dq \, d\rho_N  \\
      & = \int_{\mathbb{Q}_N} \bigg[ 
          \int_{\frac{\anu ( 1 - 2\rho_N(t)) }{1-\anu}}^\infty
	  \mathcal{N}(q; \amu, \axi) \, dq
      \bigg] \bar{P}(\rho_N,t) \, d \rho_N \\
     & = \int_{\mathbb{Q}_N} 
       \frac{1}{2} 
       \erfc \bigg[ \frac{1}{\axi \sqrt{2}} \bigg(
	 \anu\frac{1-2\rho_N(t)}{1-\anu} -\amu 
       \bigg) \bigg] \bar{P}(\rho_N,t) \, d\rho_N \\
     & = \E \bigg[ 
     \frac{1}{2} 
       \erfc \bigg[ \frac{1}{\axi \sqrt{2}} \bigg(
	 \anu\frac{1-2\rho_N(t)}{1-\anu} -\amu 
       \bigg) \bigg] \bigg] \\
      & := \E[ \Psi( \rho_N(t); \anu, \amu, \axi) ],
  \end{split}
  \label{eq:Erho}
  \end{equation}
where the expectation is taken over all possible values of $\rho_N$.

 The equation above does not close at the level of $\E[ \rho_N ]$,
  since $\Psi$ is a nonlinear function of $\rho_N$, and therefore $\E[\Psi(\rho_N)] \neq \Psi(\E[\rho_N])$.
  However, in the
  limit as $N \to \infty$, we can perform a Taylor expansion of $\psi(\rho)$
  around $\E[\rho]$, (see \cite{Helbing2011r,Weidlich2000y})
  \begin{equation}
  \begin{split}
  \E[\Psi(\rho_\infty)] & \approx \E\Big[ 
                      \Psi(\E[\rho_\infty]) + 
                      \Psi^\prime(\E[\rho_\infty])( \rho_\infty - \E[\rho_\infty]) +
		      \frac{1}{2}\Psi^{\prime\prime}(\E[\rho_\infty])(
		      \rho_\infty - \E[\rho_\infty])^2
		      \Big] \\
                 & = \E[ \Psi(\E[\rho_\infty])] + 
		      \frac{1}{2}\Psi^{\prime\prime}(\E[\rho_\infty])\E[ (
		      \rho_\infty - \E[\rho_\infty])^2] \\
		 & = \Psi(\E[\rho_\infty]) + \mathcal{O}( \var[\rho_\infty] ),
  \end{split}
  \label{eq:Taylor}
  \end{equation}
  which combined with~\eqref{eq:Erho} proves the assertion.
\end{proof}
\begin{rem}
    In Lemma~\ref{lem:coarseMap} we assume that the deterministic coupling
    constant $\anu$ is strictly between $0$ and $1$, in order to exclude
    trivial dynamics. If $\anu = 0$, then~\eqref{eq:Erho} gives
    \[
      \E[\rho_N(t+1)] = \frac{1}{2} \erfc \bigg( -\frac{\amu}{\axi \sqrt{2}}
      \bigg), \qquad t \in \Z_+
    \]
    that is, a microscopic equilibrium is reached after one time step and the
    corresponding macroscopic equilibrium does not depend upon initial
    conditions. This is to be expected, since $\anu=0$ means that agents
    disregard information about their neighbours, therefore initial conditions
    are not relevant to their choice.

    On the other hand, if $\anu = 1$ we cannot directly apply~\eqref{eq:Erho}.
    However, we have
    \[
    u_n(1) = 
    \begin{cases}
      1 & \text{\textnormal{if $ \sum_n u_n(0) > N/2$,}} \\
      0 & \text{\textnormal{otherwise,}}
    \end{cases} \qquad
      n = 1,\ldots,N,
    \]
    and so the system achieves a microscopic locked-in equilibrium after one time
    step. If, as was done in the numerical experiments of
    Figure~\ref{fig:homStates}, the microscopic initial conditions are independent
    identically-distributed variables, $u_n(0) \sim \mathcal{B}(u,0.5)$ for all
    $n=1,\ldots,N$, we have $\sum_{n=1}^N u_n(0) \sim \binomial(N,0.5)$,
    therefore
    \[
     \E [ \rho(1) ] = 1 - \prob \bigg[ \sum_{n=1}^N u_n(0) \leq N/2 \bigg]
		    = 1 - \frac{1}{2^N} \sum_{i=0}^{\lfloor N/2 \rfloor} \binom{N}{i}.
    \]
  \end{rem}

Lemma~\ref{lem:coarseMap} suggests a simple way to derive a coarse evolution map: if
the hypotheses of the lemma hold true and we are in the limit of
infinitely many agents, we can choose $U=\E[\rho_\infty] = \E[u_n]$ as our coarse
variable; then, to leading order, we obtain
\begin{equation}
U(t + 1) =  
   \frac{1}{2} \erfc \bigg[ \frac{1}{\axi \sqrt{2}} \bigg(
       \anu\frac{1-2U(t)}{1-\anu} -\amu \bigg) \bigg]
        := \Phi_\textrm{a}(U(t); \anu, \amu,\axi)
  \label{eq:coarse1DMap}
\end{equation}
It is clear that this is only an approximate evolution map, as we have
tacitly assumed that the probability distribution for $\rho_N(t)$ is unimodal
and sharply peaked, so that $\rho_\infty \approx \E[\rho_\infty]$ and
$\var[\rho_\infty] \approx 0$. The numerical simulations of
Section~\ref{sec:simulations} (in particular Figure~\ref{fig:slaving}) show that
this is a valid approximation on sufficiently short time scales. By analogy
with~\cite{Barkley2006y}, we expect that fixed points
of this first-moment map will inform us about metastable homogeneous
states of the full lock-in model, hence we proceed to discuss fixed points of the map
and their stability.

For simplicity, let us consider the case of equally-perceived products, such that $\amu = 0$, and
fixed standard deviation $\axi$ and study fixed points $U_*$ of $\Phi_\textrm{a}$
as $\anu$ is varied. For all $\anu$, the map possesses a fixed point at $U_*=1/2$,
the mixed state, which is stable for $\anu < \anu_\textrm{c}$, where
$\anu_\textrm{c}$ is computed as
\begin{equation}
\Phi'_\textrm{a}(1/2; \anu_\textrm{c}, 0,\axi) = 1 \quad \Rightarrow \quad
  \anu_c = \frac{1}{1 + \axi^{-1}\sqrt{2/\pi}}.
\label{eq:muCrit}
\end{equation}
At the critical point, two new fixed points arise (corresponding to equilibria
with increasingly high proportions of one product over the other), while the mixed
state becomes unstable at a pitchfork bifurcation. Since we have an analytic
expression for $\Phi_\textrm{a}$, we can readily apply numerical continuation
techniques and obtain the bifurcation diagram shown in
Figure~\ref{fig:analyticContinuation}.
\begin{figure}
  \centering
  \includegraphics{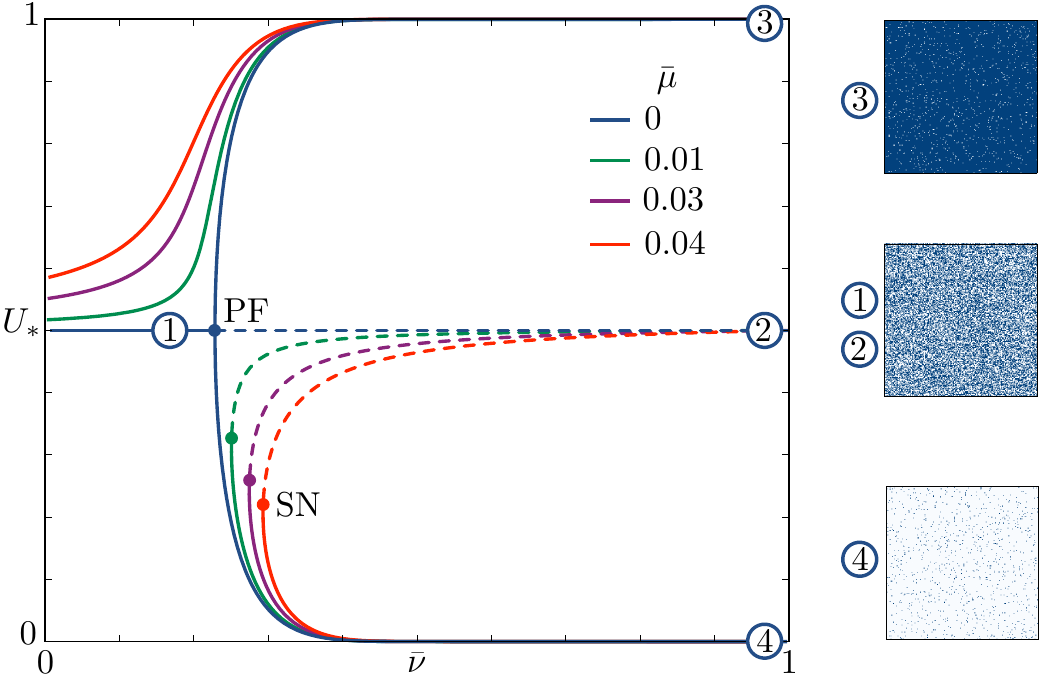}
  \caption{Fixed points of the approximate coarse map $\Phi_\textrm{a}$ as a function
  of the coupling parameter $\anu$ for $\axi=0.236$ and various values of
  $\amu$. Solid (dashed) lines represent stable (unstable) coarse equilibria.
  Representative microscopic solutions are plotted on the right. If agents are
  unbiased on average ($\amu=0$), the mixed state becomes unstable at a
  pitchfork bifurcation (PF), attained at the critical value given by
  \eqref{eq:muCrit}, and two locked-in states emerge. If agents have an average
  bias towards product 1 ($\amu > 0$), the pitchfork breaks down, giving rise to a
  saddle node bifurcation (SN). A similar scenario occurs if $\amu < 0$ (not shown).}
  \label{fig:analyticContinuation}
\end{figure}
 
Similar considerations are valid if we assume that the population has a bias
towards one of the products ($\amu \neq 0$).  Then,
the pitchfork breaks into two branches: one of them, corresponding to the product
with higher average perceived quality, is stable for all values of the coupling
$\anu$, whereas the other one destabilises at a saddle-node bifurcation. As expected,
the basin of attraction of the locked-in state is larger for the product with a
higher perceived quality.

To understand better the relation between fixed points of $\Phi_\textrm{a}$ and
metastable states of the lock-in model we refer to Figures~\ref{fig:homStates}
and~\ref{fig:slowManifold}. In the bistable region of parameter space, each
realization of the stochastic process evolves to a state that corresponds to one of
the stable fixed points of the coarse evolution map over reasonably short time
scales; these equilibria in Figure~\ref{fig:homStates} correspond to the fixed points
of the coarse evolution map on Figure~\ref{fig:analyticContinuation}.  For a
microscopic stochastic simulation starting close to such an equilibrium, this implies
that there is a distribution of 
mean purchases, unimodal and sharply peaked around
this population average, as can be observed in Figure~\ref{fig:slowManifold}.  The
equilibria computed from the analytic coarse evolution map approximate this
distribution using a Dirac-delta distribution.  However, over long time scales, both
metastable states are equally likely to occur, as Figure~\ref{fig:slowManifold}
shows.

\section{Equation-free Newton-Krylov method}\label{sec:eq-free}

In this section, we aim to obtain a numerical closure relation for the
evolution of the 
$x$-dependent macroscopic state
$\vect{U}=(U_n)_{n=1}^N$. In this case, an analytical closure approximation is
no longer valid. We thus propose an equation-free method.  We first outline the
general principle of the equation-free methodology
(Section~\ref{sec:ef-principle}). Next, we describe the concrete lifting and
restriction operators that will be used (Section~\ref{sec:ef-lift}). The main
algorithmic contribution of the present paper is the introduction of a weighted
lifting and restriction operator that allows the accurate computation of
Jacobian-vector products, as will be discussed in Section~\ref{sec:ef-Jv}.

\subsection{Principle}\label{sec:ef-principle}

As was shown in the previous sections, the lock-in model~\eqref{eq:phiMap}
consists, at the microscopic level, of individual
agents whose state keeps evolving, owing to the probabilistic nature of their
choices. Nevertheless, at the macroscopic level,  the
ensemble average \eqref{eq:macro_var} is seen to evolve to a metastable equilibrium.
In this paper, we are interested in performing a bifurcation analysis at the
macroscopic level, at which an exact, closed model is not available. The equation-free framework was developed for such tasks
\cite{Kevrekidis2003r,Kevrekidis2009q}.

The main building block in an equation-free method is the \emph{coarse
time-stepper}, which allows the performance of time-steps at the macroscopic level
(defined by \eqref{eq:macro_evolution}), using only the
simulation of $M$ realizations of the microscopic model \eqref{eq:phiMap}.  To
achieve this, the procedure relies on the definition of two operators
(\emph{lifting} and \emph{restriction}) that relate the microscopic and
macroscopic levels of description. The lifting operator maps a macroscopic state
to a microscopic one, that is, starting from a macroscopic state $\vect{U}$ and
macroscopic parameters $\vect{\Gamma}$, it generates an ensemble,
\begin{equation}\label{eq:ef-principle-ensemble}
\vect{B}=\left[\vect{u}^m\right]_{m=1}^{M} \in \BB^{N \times M},
\end{equation}
of $M$ realizations $\vect{u}^m$ ($m = 1,\ldots, M$) of the microscopic state (each
consisting of $N$ individual agents) from the ensemble average $\vect{U} \in
\RR^{N}$, as well as a set of $Q$ microscopic parameters for each agent,
$\vect{A}=[\vect{\gamma}^m]_{m=1}^M$, with $\vect{\gamma}^m\in \RR^{Q\times
N}$, sampled from the parameter distributions \eqref{eq:parDistr}
specified by the macroscopic parameters $\vect{\Gamma} = (\amu,\dmu,\alpha,\axi,\anu,\azeta)
\in \RR^P$. 
We remark that, once sampled, the agents' parameters are kept fixed
throughout the evolution step.

When generating random realizations for the microscopic state and
parameters, applying the same lifting operator multiple times will give
different results, depending on the precise random numbers that were generated
during the process.  We will denote this set of random numbers by
$\vect{\omega}\in\vect{\Omega}$, in which the sample space $\vect{\Omega}$
represents all possible sets of random numbers that can be generated; it may be
convenient to think of $\vect{\omega}$ as the set of seeds of all the
random number generators involved. This leads to an operator of the form:
\begin{equation}
  \begin{aligned}
\L \colon  \RR^N \times \RR^{P}\times \vect{\Omega} & \longrightarrow  \BB^{N \times M} \times \RR^{Q \times N \times M },\\
(\vect{U},\vect{\Gamma},\vect{\omega}) & \longmapsto \big( \vect{B}, \vect{A} \big).
  \end{aligned}
\end{equation}
In the remainder, we will also denote the lifting by 
\begin{equation}\label{eq:ef-principle-lift}
  (\vect{B},\vect{A})= \big( \L_u(\vect{\omega}) \vect{U},
  \L_\gamma(\vect{\omega}) \vect{\Gamma}\big) = \L(\vect{\omega})(\vect{U},\vect{\Gamma}),
\end{equation}
to emphasize that $\vect{\omega}$ only appears as a parameter.
\begin{rem}[Dependence on the random event $\vect{\omega}$]
The explicit introduction of the parameter $\vect{\omega}$ may seem elaborate at
first. Nevertheless, in the remainder of the text, especially when discussing
the computation of variance-reduced Jacobian vector products in
Section~\ref{sec:ef-Jv}, this notation will prove to be indispensable.
\end{rem}
\begin{rem}[Closure approximation]
The microscopic realizations have to be sampled from a probability
distribution $P(\vect{u},t)$ that is consistent with $\vect{U}(t)$, that is, we require
$\int_{\BB^{N}}\vect{u}\,P(\vect{u},t)\,d\vect{u}=\vect{U}(t)$. At this
point, we have not yet specified what probability distribution $P(\vect{u},t)$
will be used to this end. Choosing $P(\vect{u},t)$ amounts to enforcing a
closure approximation. In Section~\ref{sec:ef-lift}, we will construct several
lifting operators that perform this closure approximation numerically.
\end{rem}

Conversely, the restriction operator maps a microscopic state to a macroscopic
one, that is, it computes an appropriate ensemble average $\vect{U}\in \RR^N$ of
the $M$ realizations $\vect{B}=\left[\vect{u}^m\right]_{m=1}^{M}\in \BB^{N\times
M}$:
\begin{equation}
\begin{aligned}
  \R \colon  \BB^{N \times M} & \longrightarrow \RR^N,\qquad 
                            \vect{B} \longmapsto \vect{U}.
\end{aligned}
\label{eq:ef-principle-restrict}
\end{equation}
As a general principle, one expects the macroscopic state to be unchanged when
performing lifting followed by restriction, that is,
\begin{equation}\label{eq:ef-principle-identity}
	\R \circ \L_u \equiv \Id.
\end{equation}
In general, however, $\L_u \circ \R \neq \Id$, since it is impossible to recover exactly the microscopic information during lifting that was discarded during restriction. For the problem considered here, even
ensuring \eqref{eq:ef-principle-identity} is nontrivial, because one cannot
represent every possible value of $\vect{U}$ exactly as the ensemble average of
$M$ microscopic realizations. Specific operators that circumvent this problem
are proposed in Section~\ref{sec:ef-lift}.

Once lifting and restriction operators have been constructed, a coarse
time-stepper $\Phi_T^M$ to evolve the macroscopic state $\vect{U}$ over a time
interval of length $T$ is constructed as a three-step-procedure
(lift--evolve--restrict), in which the microscopic evolution is simulated
independently for each of the $M$ realizations, i.e.,
\begin{equation}\label{eq:ef-principle-cts}
  \vect{U}(t+T)=\Phi_T^M(\vect{\omega}) ( \vect{U}(t); \vect{\Gamma}) = ( \R \circ \EE_T \circ \L(\vect{\omega}) )
  \big( \vect{U}(t) ; \vect{\Gamma} \big),
\end{equation}
with $\L(\vect{\omega})$ and $\R$ defined in \eqref{eq:ef-principle-lift} and
\eqref{eq:ef-principle-restrict}, and $\EE_T$ defined as 
\begin{equation}
\begin{aligned}
  \EE_T \colon  \BB^{N \times M} \times \RR^{Q\times N \times M} 
                         & \longrightarrow \BB^{N \times M}, \;\;
                   (\vect{B},\vect{A}) &\longmapsto 
			    \big[\varphi_T(\vect{u}^m;\vect{\gamma}^m)\big]_{m=1}^M,
\end{aligned}
\label{eq:ef-principle-evolve}
\end{equation}
where we have denoted by $\varphi_T$ the $T$th iterate of the lock-in
map~\eqref{eq:phiMap}. Note that, in the limit $M \to \infty$, the coarse
time-stepper approaches
\begin{equation} \label{eq:cts_limit}
  \vect{U}(t+T)=
  \int_{\BB^{N}}\vect{u}\left[\int_{\RR^Q} \int_{\BB^{N}}\Psi_T(\vect{u}|\vect{v},\vect{\gamma}) 
  p(\vect{v},t,\vect{\gamma}| \vect{U}(t), \vect{\Gamma}) \, d\vect{v}d\vect{\gamma}\right]d\vect{u},
\end{equation} 
in which we have introduced the transition kernel $\Psi_T$ over a time interval
$T$ and the probability distribution
$p(\vect{v},t,\vect{\gamma}|\vect{U}(t),\vect{\Gamma})$ conditioned
upon $(\vect{U}(t),\vect{\Gamma})$, from which the samples are taken. 
The interpretation of the coarse time-stepper as a numerical closure follows by
comparing equation~\eqref{eq:cts_limit} with~\eqref{eq:macro_evolution}, and
noticing that the right-hand side is completely determined by $\vect{U}(t)$ and $\vect{\Gamma}$, since the probability distribution $p(\vect{v},t,\vect{\gamma}|\vect{U}(t),\vect{\Gamma})$ is conditioned upon $\vect{U}(t)$ and $\vect{\Gamma}$. 

If the system \eqref{eq:phiMap} possesses macroscopic steady states, these can be
found (for fixed macroscopic parameters $\vect{\Gamma}=\vect{\Gamma}_*$) by solving the
nonlinear system,
\begin{equation}\label{eq:ef-principle-newton}
  \vect{F}(\vect{U}_*) = 
\vect{U}_* -
\Phi^M_T(\vect{\omega})(\vect{U}_*,\vect{\Gamma}_*) = 0,
\end{equation}
for an appropriate choice of $M$ and $T$.
This procedure allows the computation of unstable steady states that would
not be reached by direct simulation. By adding a pseudo-arclength condition, one
can also perform continuation to obtain a branch of steady states as a function
of a free parameter.

In each Newton iteration, one needs to solve a linear system involving the
Jacobian of $\Phi_T^M$, denoted as $D\Phi_T^M(\vect{U};\vect{\Gamma})$. Since we do not have
an explicit formula for $D\Phi_T^M(\vect{U};\vect{\Gamma})$, we are forced to
use an iterative method (such as GMRES) that only requires Jacobian-vector
products, and to estimate such Jacobian-vector products using a finite
difference approximation.  However, we recall that, for a finite number of
realizations $M$, the coarse time-stepper $\Phi_T^M$ is stochastic.  Hence,
repeating the same coarse time-step with two sets of random numbers
$\vect{\omega}_{1,2}$ gives different results.  A standard Monte Carlo argument
\cite{caflisch1998monte} reveals that
\[\var \left[ \Phi_T^M(\vect{\omega}_1)(\vect{U};\vect{\Gamma}) -
\Phi_T^M(\vect{\omega}_2)(\vect{U};\vect{\Gamma}) \right] \le
C\dfrac{1}{M},  
\]
resulting in typical deviations of $\mathcal{O}(1/\sqrt{M})$.
Then, estimating Jacobian-vector products using the simple finite-difference
formula
\begin{align}
  D\Phi_T^M(\vect{U};\vect{\Gamma}) \vect{V} &\approx
  \frac{\Phi_T^M(\vect{\omega}_2)(\vect{U} +
  \varepsilon \vect{V} ;\vect{\Gamma}) -%
  \Phi_T^M(\vect{\omega}_1)(\vect{U} ;\vect{\Gamma})}{\varepsilon}, \\
&\approx
  \frac{\Phi_T^M(\vect{\omega}_2)(\vect{U} ;\vect{\Gamma}) + \varepsilon D\Phi_T^M(\vect{\omega}_2)(\vect{U};\vect{\Gamma}) \vect{V}-%
  \Phi_T^M(\vect{\omega}_1)(\vect{U} ;\vect{\Gamma})}{\varepsilon}, 
\label{eq:ef-principle-jacFD}
\end{align}
with $\varepsilon\ll 1$ will result in an $\mathcal{O}(1/(\varepsilon^2
M))$ variance.

Consequently, the variance of $D\Phi_T^M(\vect{U};\vect{\Gamma})
\vect{V}$ 
will grow unboundedly as $\varepsilon$ tends to zero.
One should therefore
aim at using the same random numbers twice, both with the unperturbed and
perturbed initial conditions. 
A method to enforce
the use of the same random numbers is proposed in Section~\ref{sec:ef-Jv}.

\subsection{Lifting and restriction}\label{sec:ef-lift}

In this section, we describe two lifting operators, as well as their corresponding
restriction operator. For both approaches, the microscopic parameters are initialized by
generating i.i.d.~samples for each agent in each realization from the governing
probability distributions~\eqref{eq:parDistr}. The difference between both
lifting and restriction operators is limited to the initialization of the
microscopic state.  We emphasize as well that, for each realization, the
microscopic state is initialized independently of the parameter values. 

\subsubsection{Simple lifting and restriction}

Let us first describe a simple approach. We are given a macroscopic state
$\vect{U}=(U_n)_{n=1}^N$ and we want to generate $M$ realizations of $N$ agents,
consistently with that macroscopic state. To create these microscopic
realizations $\vect{B}=[\vect{u}^m]_{m=1}^M$, with $\vect{u}^m=(u_n^m)_{n=1}^N$,
we can sample, at each 
$x_n$, 
the Bernoulli distribution with
mean $U_n$, that is
\begin{equation}\label{eq:naive-lifting}
    u_n^m \sim \B(u;U_n) \iff
    \begin{cases}
     \Pr(u_n^m = 1) = U_n \\
     \Pr(u_n^m = 0) = 1-U_n,
    \end{cases}
    \quad m = 1,\ldots,M 
\end{equation}
Combining this sampling of the microscopic state with a sampling procedure for
the microscopic parameters of the individual agents, we obtain a lifting operator
$\L(\vect{\omega})$
of the type~\eqref{eq:ef-principle-lift}.
The corresponding restriction operator is then given by taking the empirical
average over the set of $M$ realizations, 
\begin{equation}
\begin{aligned}
  \R \colon  \BB^{N \times M} & \longrightarrow \RR^N,\qquad \vect{B} &
  \longmapsto \vect{U}=\frac{1}{M} \sum_{m=1}^M \vect{u}^m.
\end{aligned}
\label{eq:ef-naive-restrict}
\end{equation}

The simple lifting and restriction operators defined above cannot
satisfy the consistency condition~\eqref{eq:ef-principle-identity} for an
arbitrary value of $\vect{U}(t)$, since the restriction can only map onto
$\mathbb{Q}^N_M = \mathbb{Q}_M \times \dots \times \mathbb{Q}_M$ instead of onto
$\RR^N$, i.e., only integer fractions of $M$ can be represented. The
incurred
discrepancy is
essentially a sampling error, since the sampling procedure outlined above only
ensures~\eqref{eq:ef-principle-identity} in the limit $M \to \infty$. Indeed,
when $\left(\vect{B},\vect{A}\right)=\L(\vect{U},\vect{\Gamma})$, with
$\vect{B}=\left[\u^m\right]_{m=1}^M$, then 
\begin{equation} 
  \lim_{M \to \infty}\frac{1}{M} \sum_{m=1}^M \vect{u}^m = \vect{U}.
  \label{eq:ef-naive-restrict-limit} 
\end{equation}

\subsubsection{Weighted lifting and restriction}

The main idea of the present paper, which is key to all the numerical methods
that follow, is the introduction of a new restriction operator that replaces
the empirical average~\eqref{eq:ef-naive-restrict} by a weighted average of
the form
\begin{equation}
  \R_{\textrm{w}}(\vect{w}) \colon  \BB^{N \times M} \longrightarrow \RR^N,\qquad 
		\vect{B} \longmapsto \vect{U}=\frac{1}{M} \sum_{m=1}^M w^m\vect{u}^m,
\label{eq:ef-weighted-restrict}
\end{equation}
in which $\vect{w}=\left[w^m\right]_{m=1}^{M}\in \RR^M$ is a vector of weights
satisfying
\begin{equation}\label{eq:ef-weights-pdf}
  \frac{1}{M}\sum_{m=1}^M w^m=1.  
\end{equation}
The restriction operator $\R_{\textrm{w}}$ is specified completely only once
the weights $\vect{w}$ are known; they will be selected such that the
constraint~\eqref{eq:ef-principle-identity} is satisfied
exactly, which implies that the restriction operator will depend on the
specific realizations $\vect{u}^m$ that were generated during the lifting.

Before outlining the procedure, let us highlight the rationale behind
the introduction of the weighted average. As noted in
Remark~\ref{rem:mc_weights}, the probability distribution $P(\vect{u},t)$ can
be discretized according to two guiding principles: (i) \emph{deterministically},
that is, we consider every possible realization and attach to it a probability
weight expressing how likely the realization is to occur, which results in
equation~\eqref{eq:int_bernouilli}; or (ii) \emph{stochastically}, that is, we
sample a finite number of realizations from the corresponding probability
distributions, resulting in the estimate~\eqref{eq:int_mc}. Option (i)
is unfeasible because it requires considering $M=2^N$ realizations (many of
which will be extremely unlikely), while option (ii) will contain a sampling
error such that the identity~\eqref{eq:ef-principle-identity} is violated.
Introducing the weighted restriction~\eqref{eq:ef-weighted-restrict} can then be
seen as a hybrid approach that allows
satisfying~\eqref{eq:ef-principle-identity} with a limited number of
realizations $M\ll 2^N$; the condition~\eqref{eq:ef-weights-pdf}
ensures that the weight $w^m$, attached to the realization $\vect{u}^m$, can be
interpreted as the probability of obtaining that realization out of all the realizations in the sample. This interpretation also imposes the condition that all the weights be positive. We shall see
there is an interplay between the creation of the $M$ realizations and the
computation of the corresponding weights for the restriction.

A possible way to compute weights is the following: first, we generate $M'$
realizations $\vect{u}^m$ according to the naive procedure~\eqref{eq:naive-lifting};
since we know
that this procedure yields
the desired result as $M'$ tends to infinity, it seems reasonable to attach weights
that are as close to 1 as possible, while satisfying the
identity~\eqref{eq:ef-principle-identity}, as well as the
constraint~\eqref{eq:ef-weights-pdf}. As will become clear further on, this
procedure will turn out to allow for optimization problems that are either
unfeasible (with no possible solutions) or ill-posed (with infinitely many
possible solutions).  To see this, we formulate the following least squares
problem,
\begin{align}
  &\vect{w} = \arg\min \frac{1}{2} \sum_{m=1}^{M'} \left(w^m -
  1\right)^2, \label{eq:min1} \\ 
  &\frac{1}{M'} \sum_{m=1}^{M'} w^m\vect{u}^m=\vect{U}, \label{eq:min2} \\ 
  & \frac{1}{M'}\sum_{m=1}^{M'} w^m=1,\label{eq:min3} \\
&w^m \geq 0 \qquad 1\le m \le M'. 
\end{align} 
We recall here a basic result in minimization
problems~\cite{Gould2001a,Madsen2004y}:
\begin{lem}\label{lem:minProb}
  Let us consider the following equality-constrained quadratic minimization problem
  \begin{align*}
    & \vect{w} = \arg \min \frac{1}{2} \vect{w}^T\vect{w}
    - \vect{g}^T \vect{w}, \\
    & \vect{C} \vect{w} = \vect{b},
  \end{align*}
  where $\vect{w},\vect{g} \in \RR^M$, $\vect{b} \in \RR^N$ and where $\vect{C} \in
  \RR^{N \times M}$, with $N<M$, is a constraint matrix with full rank, then
  \begin{equation}
  \begin{bmatrix}
    \vect{I} & \vect{C}^T \\
    \vect{C} & \vect{0}
  \end{bmatrix}
  \begin{bmatrix}
    \vect{w} \\
    \vect{\lambda}
  \end{bmatrix}
  =
  \begin{bmatrix}
    \vect{g} \\
    \vect{b}
  \end{bmatrix},
  \label{eq:weightsLinSys}
  \end{equation}
  where $\vect{\lambda} \in \RR^{N}$ is the associated Lagrange
  multiplier. The linear system~\eqref{eq:weightsLinSys} has a unique solution.
\end{lem}
Let us now consider the difficulties that may lead to a rank-deficient constraint matrix $\vect{C}$:
\begin{enumerate}
  \item
    The sampling procedure~\eqref{eq:naive-lifting} can yield multiple identical
    realizations of the microscopic state (identical columns in the constraint
    matrix). For instance, this may happen with high probability,
    if the macroscopic state $\vect{U}$ is close to $0$ or $1$ for all agents, such
    that all realizations consists of almost all $0$ or all $1$). 
  \item
    The sampling procedure can also yield repeated rows in the constraint matrix,
    when two agents ($n_1$ and $n_2$) have an identical state
    in each of the realisations, that is, $u^m_{n_1}=u^m_{n_2}$ for all $m$. This also
 happens with high probability if the macroscopic state $U_n$ is close to
    $0$ or $1$ for two or more agents.  
 \item
 For a given agent $n$, one might find that all realizations have the same
   value ($0$ or $1$). When $U_n$ is not identically $0$ or $1$, this leads to an infeasible constraint; again, this situation is likely to occur
   when the macroscopic state $U_n$ for some agent $n$ is close to $0$ or $1$.
\end{enumerate}
To circumvent these problems, we will discard duplicate realizations during the
computation of the weights and extend the
sample set with artificially created samples. To minimise perturbations with respect
to the underlying probability distributions, the target weights $\vect{g}$ will be
adjusted accordingly.

To be specific, we  circumvent the
first problem as follows. We denote by $\vect{B}'$ the ensemble of $M'$
realizations that were generated with the procedure~\eqref{eq:naive-lifting},
and write this ensemble in a different representation 
$(\tilde{\vect{B}}, \vect{g})$
where we only retain \emph{unique} realizations, as well as their \emph{cardinality}
in the ensemble
$\vect{B}'$,
\begin{equation}
  \tilde{\vect{B}}=\left[\vect{u}^m\right]_{m=1}^{\tilde{M}}, \qquad
  \vect{g}=\left[g^m\right]_{m=1}^{\tilde{M}},
\end{equation} 
where
\begin{equation}
  g^m=\#\{\u^m\in\vect{B}'\} \quad \text{for all $m$ such that  $\vect{u}^m \in
  \tilde{\vect{B}}$}.
\end{equation}
We note that, by definition, we have $\sum_{m=1}^{\tilde{M}}g^m=M'$.
During optimisation, we will then compute a weight for each
single realisation that is close to $g^m$ (see later), to take into account
the fact that each realisation appeared $g^m$ times in our original sampling.  Note
that afterwards we retain all individual realisations, since they will have
different values for the microscopic parameters over which we want to average.

To circumvent the second and third problem, we create artificial realizations in
the lifting step that are unlikely to be obtained by the naive
sampling procedure~\eqref{eq:naive-lifting}, and assign to them a
target weight of $0$ to minimise artefacts in the resulting probability
distributions. First, we scan the new constraint matrix
and search for duplicate rows. For each repeated row $n_\textrm{r}$, we add a
realization as follows
\begin{equation}
u_n =
\begin{cases}
  1 & \text{if $n=n_{\textrm{r}}$,} \\
  0 & \text{otherwise.}
\end{cases}
\label{eq:nr}
\end{equation}
Then, we check if there exists a row $n_0$ of 0s or a row $n_1$ of 1s and if so
add the following realizations respectively
\begin{equation}
u_n =
\begin{cases}
  1 & \text{if $n=n_0$,} \\
  0 & \text{otherwise,}
\end{cases}
\qquad
u_n =
\begin{cases}
  0 & \text{if $n=n_1$,} \\
  1 & \text{otherwise.}
\end{cases}
\label{eq:n0n1}
\end{equation}

We collect all additional realizations~\eqref{eq:nr}--\eqref{eq:n0n1} in the set
$\vect{B}''$ and compute the cardinality as follows:
\begin{equation}
  g^m=0, \qquad \text{for all $m$ such that $\vect{u}^m \in \vect{B}''$}
  \label{eq:goalZero}
\end{equation}
indicating that those realizations appear with cardinality $0$ in the original
sampling according to procedure~\eqref{eq:naive-lifting} and have only been
added to regularise the constraint matrix.

A weighted lifting operator $\mathcal{L}_{\textrm{w}}(\vect{\omega})$ is then given as the
set of $M=M'+M''$ realizations $\vect{B}=\vect{B}'\cup\vect{B}''$, along with the
correspondingly sampled microscopic parameter values $\vect{A}$, which together form
an operator of the type~\eqref{eq:ef-principle-lift}.  
The weights $\vect{w}\in \RR^{M}$ (with $M=M'+M''$) that will be used in the
restriction $\R_{\textrm{w}}$ are such that both~\eqref{eq:ef-weighted-restrict}
and~\eqref{eq:ef-weights-pdf} are satisfied, and such that the natural sampling
frequencies, as exemplified by the counters $g^m$, are matched as closely as
possible. 
We first compute weights $\tilde{\vect{w}}\in \RR^{\tilde{M}+M''}$ for all
elements of $\tilde{\vect{B}}\cup\vect{B}''$ by solving the regularised constrained
minimization problem 
\begin{align}
	&\tilde{\vect{w}}= \arg\min \frac{1}{2} \sum_{m=1}^{\tilde M + M''} \left(\tilde{w}^m - \dfrac{M}{M'}g^m\right)^2,
\label{eq:minProb} \\ 
&\frac{1}{M} \sum_{m=1}^{\tilde M + M''} \tilde{w}^m\vect{u}^m=\vect{U},
\label{eq:constr1}\\ 
& \dfrac{1}{M}\sum_{m=1}^{\tilde M + M''} \tilde{w}^m=1,\label{eq:constr2}\\
& \tilde{w}_m \geq 0, \qquad 1\le m \le \tilde M + M''. \label{eq:constr3}
\end{align} 
In the system above, we conventionally assumed that
$\vect{u}^m\in\tilde{\vect{B}}$ when $1\le m \le \tilde{M}$ and
$\vect{u}^m\in\vect{B}''$ when $\tilde{M} < m \le \tilde{M}+M''$.
The choice of the goal function~\eqref{eq:minProb} ensures that
constraints~\eqref{eq:constr1} and~\eqref{eq:constr2} are not affected by the
presence of additional realizations with weights that are identically zero, that
is, $w^m=(M/M')g^m$ is a solution that satisfies~\eqref{eq:constr2}.

We then transform the weights $\tilde{\vect{w}}\in \RR^{\tilde{M}+M''}$ back to
weights $\vect{w}\in\RR^{M}$ for the $M=M'+M''$ realizations in
$\vect{B}=\vect{B}'\cup\vect{B}''$. This is done by selecting, for each element
$\vect{u}^m\in\vect{B}'$, the (unique) index $\tilde{m}^*$ such that
$\vect{u}^{m}=\vect{u}^{\tilde{m}^*}$ with
$\vect{u}^{\tilde{m}^*}\in\tilde{\vect{B}}$, and setting
$w^m=\tilde{w}^{m^*}/g^{m^*}$ for $1 \leq m \leq \tilde M$.

\begin{rem}[Effect of regularization on probability distributions of the microscopic
  states] 
  \label{rem:regularization}
  A natural question arises as to whether the regularisation procedure
  proposed above has an impact on the probability distributions of the microscopic
  states. With the regularization, we amend the lifted realizations in two ways.
  Firstly, we remove identical realizations from the constraint matrix and we assign
  to the corresponding weight a higher target (the vector
  $\vect{g}$ in Equation~\eqref{eq:minProb} contains the cardinality of the
  unique realisations in $\vect{B}'$): with this procedure we do not alter the
  underlying probability distribution of the microscopic states, in that realisations
  that have been removed will have a correspondingly higher weight. Secondly, we add
  artificial realisations, which in principle create a bias in the underlying
  microscopic distribution: for this bias not to affect the outcome of our
  computations, the associated weights should be vanishingly small, hence we
  prescribe for them a target equal to $0$ (again via the vector $\vect{g}$) and we
  expect that these weights tend to $0$ as $M \to \infty$.
\end{rem}

\begin{rem}[Numerical solution of the minimization problem]
  We solve~\eqref{eq:minProb}--\eqref{eq:constr2} using a single Cholesky
  factorization~\cite{Madsen2004y}. Algorithms based on the Conjugate Gradient
  Method can also be employed for large equality-constrained quadratic
  problems~\cite{Gould2001a}. In our computation, we do not explicitly
  require~\eqref{eq:constr3}: positivity of the weights is assessed in a
  post-processing step, and used to determine whether enough realizations were
  taken (we increase $M$ until all weights are positive). It is also possible
  (albeit more expensive) to include the inequality constraints~\eqref{eq:constr3}
  and use iterative methods to solve the minimization
  problem~\cite{Madsen2004y}.
\end{rem}
 
\subsection{Variance-reduced Jacobian-vector products}\label{sec:ef-Jv}

Let us now discuss the Jacobian-vector multiplication that was introduced
in~\eqref{eq:ef-principle-jacFD}. As indicated before, a problem with using
Equation~\eqref{eq:ef-principle-jacFD} directly is the presence of numerical
noise, which should be addressed by using the same random numbers in both the
unperturbed and perturbed simulations. To achieve this, we use the same realizations,
microscopic parameters and random time paths in both the perturbed and unperturbed
coarse time-stepper; the only difference is in the computation of the weights. For
the perturbed coarse time-stepper, we replace the constrained optimization problem
for the weights by 
\begin{align}
&\tilde{\vect{w}}_\varepsilon= \arg\min \frac{1}{2} \sum_{m=1}^{\tilde M + M''}
\left(\tilde{w}_\varepsilon^m - \frac{M}{M'}g^m\right)^2,
\label{eq:minProb_pert} \\ 
&\frac{1}{M} \sum_{m=1}^{\tilde M + M''}
\tilde{w}^m_\varepsilon\vect{u}^m=\vect{U}+\varepsilon\vect{V},
\label{eq:constr1_pert}\\ 
& \frac{1}{M}\sum_{m=1}^{\tilde M + M''} \tilde{w}^m_\varepsilon=1,\label{eq:constr2_pert}\\
& \tilde{w}^m_\varepsilon \geq 0, \qquad 1\le m \le M, \label{eq:constr3_pert}
\end{align} 
Note that only the constraint~\eqref{eq:constr1_pert} has changed with
respect to the unperturbed optimization problem (see
equation~\eqref{eq:constr1}). Since the solution of the optimization problem
depends continuously and differentiably on the right-hand side
of the constraints, small perturbations on the right-hand side
of~\eqref{eq:constr1_pert} lead to small perturbations in weights.
Furthermore, since we are using the same microscopic realizations $\vect{u}^m$
in the constraints of the perturbed and unperturbed minimization problems, we
have effectively imposed $\vect{\omega}_1=\vect{\omega}_2$ in the-finite
difference formula~\eqref{eq:ef-principle-jacFD}, hence the variance of
$D\vect{\Phi}(\vect{\omega}_1)(\vect{U}) \vect{V}$ is bounded and of
$\mathcal{O}(1/M)$.

\begin{figure}
  \centering
  \includegraphics{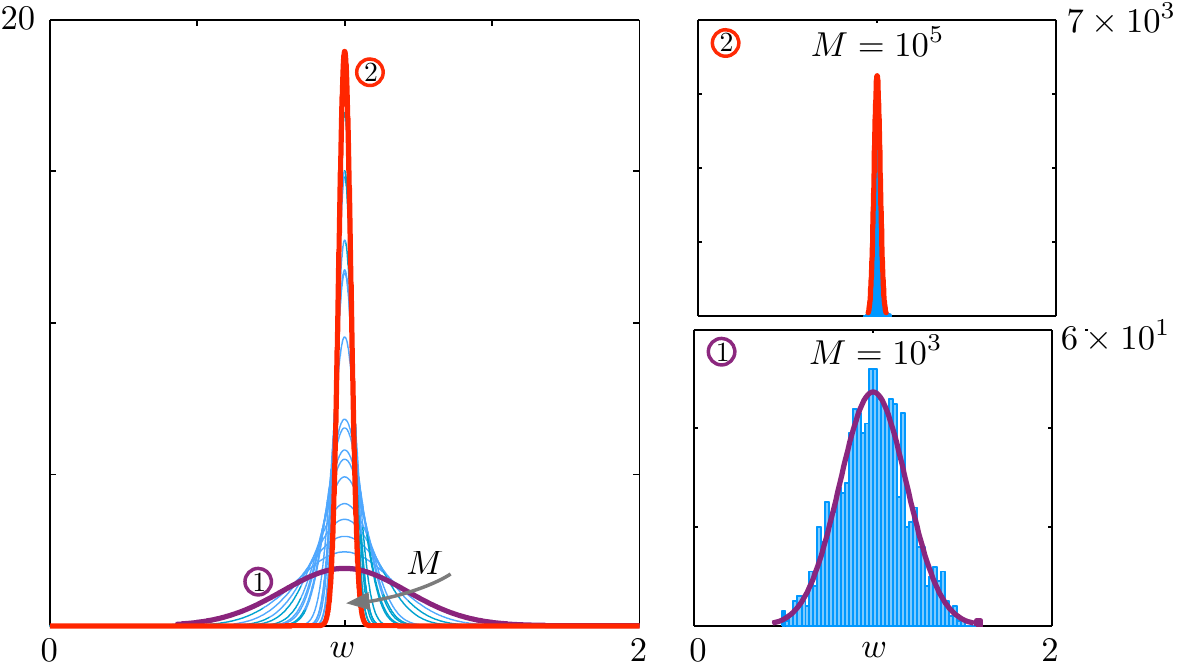}
  \caption{Distribution of weights as a function of the total number of
  realizations
  $M$. Weights are obtained by lifting a mixed state and solving the corresponding
  minimization problem for values of $M$ ranging between $10^3$ (curve $1$, magenta)
  and $10^5$ (curve $2$, red). The resulting data is fitted to a Gaussian
  distribution. As $M$ increases, the weights are sharply distributed around $1$.
  Parameters as in E1 in Table~\ref{tab:params}.}
  \label{fig:wVsM}
\end{figure}
\begin{figure}
  \centering
  \includegraphics{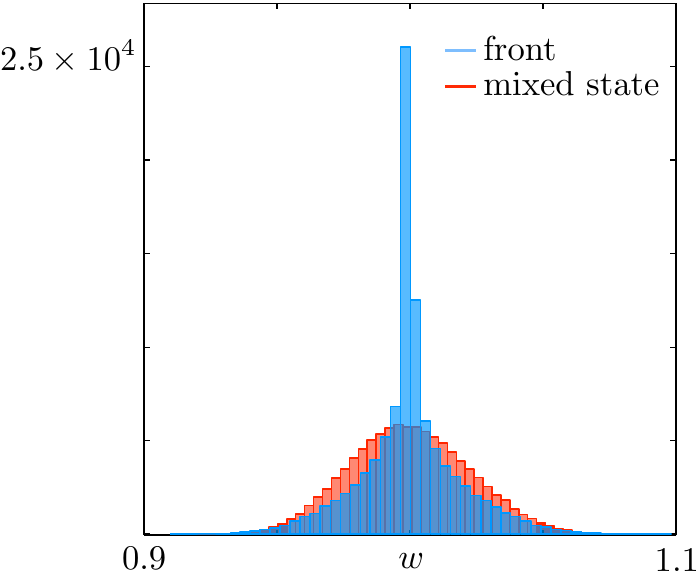}
  \caption{Distribution of weights with $M=10^5$ realizations, for two different
  macroscopic steady states: a mixed state (red) and a front (blue). Both
  distributions are sharply peaked around one, albeit the distribution for fronts is
  not Gaussian. Parameters as in E1 (mixed state) and E3 (front) in
  Table~\ref{tab:params}.}
  \label{fig:wFVsMS}
\end{figure}

In the limit of infinitely many realizations (where all weights converge to
$1$), the presented procedure converges to the exact Jacobian-vector product.
For finite values of $M$, there will be noise in the Jacobian-vector product as
a result of the random selection of a subset of all possible realizations. The
procedure only prevents noise blowup that would arise if a different selection
of realizations were considered for the perturbed and unperturbed coarse
time-step.  

\section{Numerical properties of the equation-free method}
\label{sec:numericalProperties}

In this section we show a series of numerical tests that highlight the numerical
properties of the weighted lifting and lead to an appropriate calibration of the
Newton-GMRES solver. For our tests we used a population of either $40$ or $400$
agents, a number of realizations varying between $10^3$ and $10^5$ and different
types of macroscopic steady states. Here and henceforth we will denote by
\textit{locked-in states} homogeneous macroscopic states with $U_n \approx 0$ or $U_n
\approx 1$ for all $n$, by \textit{mixed states} solutions with $U_n \approx 0.5$ for
all $n$, and by \textit{fronts} solutions that connect two locally locked-in states.
For these solutions, which were previously found via direct numerical simulations
in Figures~\ref{fig:homStates}--\ref{fig:inhomStates}, we use parameters of
E$1$--E$3$ in Table~\ref{tab:params}.  Note that, when computing fronts, we
effectively restrict our computations to one-dimensional lattices (which develop along the $x$
direction), and discard the $y$-coordinate of the lattice. We
stress that the numerical procedure presented here is unchanged in the case of
two-dimensional patterns. Unless otherwise stated, we set a time horizon
$T=20$ for the coarse time stepper.

\subsection{Convergence of the weights}

In our first numerical experiment, we fix $N=40$, lift macroscopic steady states
with the weighted operator $\L_\textrm{w}$ and plot the weight distribution as a
function of the number of realizations $M$. By construction (see
Section~\ref{sec:ef-lift}), we expect weights to be sharply distributed
around $1$ as $M$ tends to infinity. In Figure~\ref{fig:wVsM}, we lift a mixed
state for various $M$ and observe that the weight distribution is well
approximated by a Gaussian and tends to a Dirac distribution as $M \to \infty$.

Similar results are also obtained (not shown) for locked-in states and fronts.
However, we note that weights distributions associated with these states are not
necessarily Gaussian, as shown in Figure~\ref{fig:wFVsMS}. We point out that
for such macroscopic states, many weights are assigned a goal equal to $0$,
according to Equation~\eqref{eq:goalZero}. It is not surprising that the
distributions for these states, for which $M'' \gg 1$, are different to the ones
associated with a mixed state, for which $M'' \approx 0$ (see also
Remark~\ref{rem:regularization}). 

\subsection{Convergence of the Jacobian-vector product}
\begin{figure}
  \centering
  \includegraphics{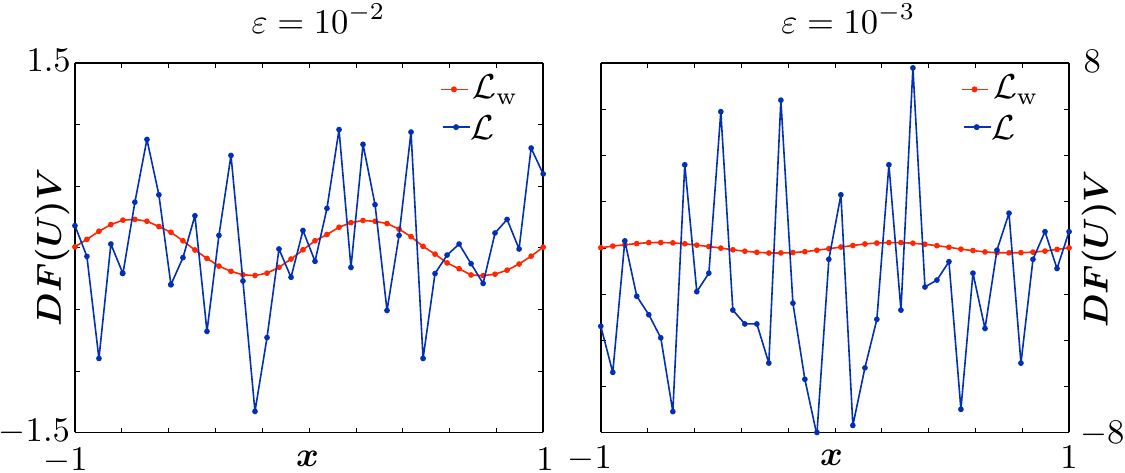}
  \caption{Jacobian-vector product $\vect{D}\vect{F}(\vect{U})\vect{V}$ of
  $\vect{F}(\vect{U}) = \vect{U} - \vect{\Phi}^M_T(\vect{U})$. We use a macroscopic front
  solution $\vect{U}$ and a sinusoidal perturbation $\vect{V}$ with $\Vert \vect{V}
  \Vert_2 = 1$. The unweighted lifting produces inaccurate Jacobian-vector product
  evaluations, whose norm becomes unbounded as we decrease the relative size of the
  perturbation $\eps$. On the other hand, weighted lifting preserves the
  structure of the perturbation. For this experiment we used $N=40$, $M=10^4$
  whereas all other parameters are chosen as in E3 of Table~\ref{tab:params}.}
  \label{fig:jacVecProdProf}
\end{figure}

We test the numerical properties of weighted Jacobian-vector products with a
second numerical experiment. We select a region of parameter space in which a
stable macroscopic front $\vect{U}$ is observed (corresponding to E3 in
Table~\ref{tab:params}) and compute a single evaluation of the Jacobian-vector
product $\vect{D}\vect{F}(\vect{U})\vect{V}$, where
$\vect{F}$ is given by~\eqref{eq:ef-principle-newton}, $\vect{D}\vect{\Phi}$ is
estimated by~\eqref{eq:ef-principle-jacFD} and $\vect{V}$ has unit norm and a
sinusoidal profile 
in $x$.
If $\vect{\Phi}(\vect{U}+\epsilon \vect{V})$ and $\vect{\Phi}(\vect{U})$ are
calculated using two independent function evaluations, the Jacobian-vector
product is severely affected by noise and completely loses the 
structure of the perturbation $\vect{V}$ (blue lines in
Figure~\ref{fig:jacVecProdProf}). Furthermore, this effect is greatly amplified
as we decrease $\varepsilon$, as the Jacobian-vector product becomes unbounded.

\begin{figure}
  \centering
  \includegraphics{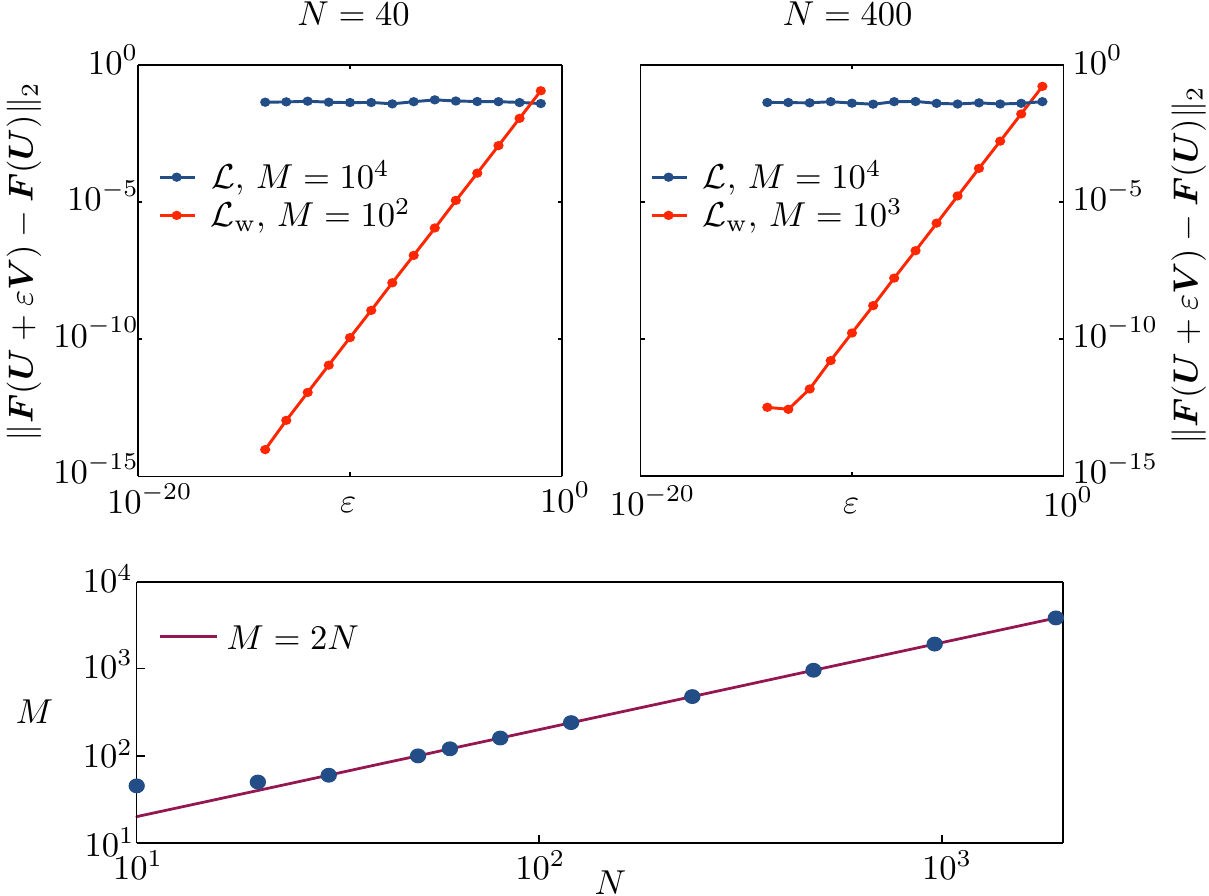}
  \caption{
  For a front $\vect{U}$ and a unit-norm random vector $\vect{V}$, we
  compute the 2-norm of $\vect{F}(\vect{U}+\eps\vect{V}) -\vect{F}(\vect{U})$
  as a function of $\varepsilon$, using weighted and unweighted lifting
  operators, for $\vect{F}(\vect{U}) = \vect{U} - \Phi^M_T(\vect{U})$. Top-left: if
  we set $N=40$ and use $10^4$ unweighted realizations, noise affects the
  evaluation of the Jacobian action, whereas $10^2$ weighted realizations are
  sufficient to obtain an $\mathcal{O}(\varepsilon)$ curve. Top-right: the
  experiment is repeated for $N=400$. 
  Bottom: we repeat the computations in the top
  panels for various values of $N$ and show the number of weighted realisations
  which consistently give an $\mathcal{O}(\varepsilon)$ curve in the Jacobian
  evaluation. The experiment shows that weights become effective with
  $M=\mathcal{O}(N)$ realizations.
  }
  \label{fig:jacVecProd}
\end{figure}

On the other hand, using weighted operators and the variance-reduced Jacobian-vector
product outlined in Section~\ref{sec:ef-Jv}, we maintain the 
structure of the perturbation and the Jacobian-vector product varies smoothly as a
function of   $x$. A further confirmation is found in Figure~\ref{fig:jacVecProd}, where
we plot the $2$-norm of $\vect{F}(\vect{U}+\eps\vect{V}) -\vect{F}(\vect{U})$ as a
function of $\eps$. In particular, we seek the minimum number of
realizations required to obtain smooth Jacobian evaluations, that is, an
$\mathcal{O}(\eps)$ curve: if $N=40$, then $100$ weighted realizations
are sufficient to obtain a smooth Jacobian evaluation, whereas $10000$
unweighted realizations are still affected by noise (left panel of
Figure~\ref{fig:jacVecProd}). If we increase the system size to $N=400$, then
$1000$ weighted realizations are sufficient to observe an $\mathcal{O}(\eps)$
curve. The experiments in Figure~\ref{fig:jacVecProd} show that weights are effective
with $M = \mathcal{O}(N)$ realizations.
\begin{figure}
  \centering
  \includegraphics{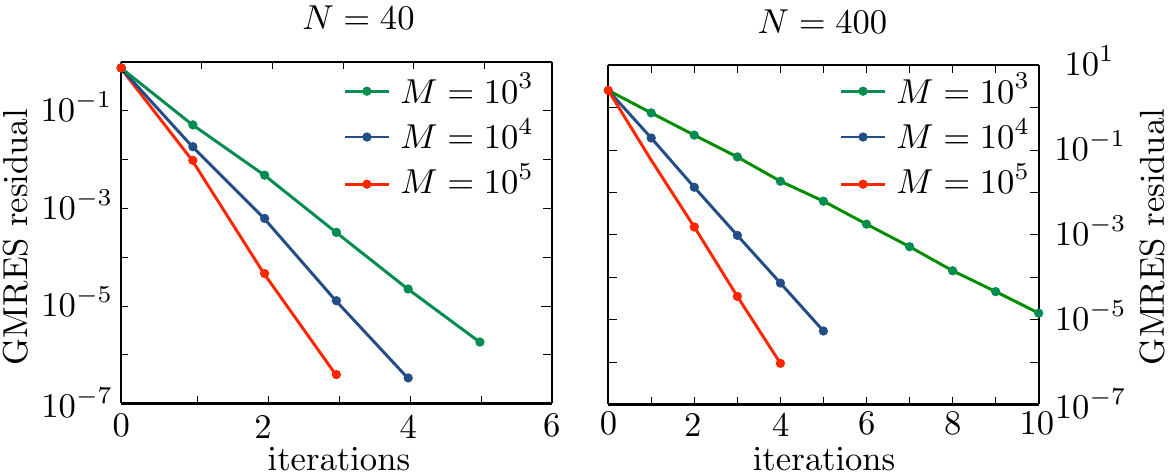}
  \caption{Convergence of GMRES iterations for the problem 
  $ \vect{D}\vect{F}(\vect{U}) \vect{V} = -\vect{F}(\vect{U})$, where 
  $\vect{U}$ is a macroscopic front and $\vect{D}\vect{F}(\vect{U}) = \vect{I} -
  \vect{D}\vect{\Phi}^M_T(\vect{U})$, for different numbers of agents $N$ and numbers of
  realizations $M$. Jacobian actions are computed with weighted lifting and
  $\varepsilon = 10^{-5}$.}
  \label{fig:GMRESRes}
\end{figure}
\subsection{Convergence of GMRES iterations}
\label{subsec:GMRESConvergence}
The next step towards the construction of our Newton-Krylov solver is the solution of
the linear system associated with the Jacobian $\vect{D}\vect{F}$ of
$\vect{F}$. We use GMRES to solve iteratively the system
$ \vect{D}\vect{F}(\vect{U}) \vect{V} = -\vect{F}(\vect{U}) $
where $\vect{U}$ is a mixed state and $\vect{D}\vect{F}$ is computed using weighted
operators and variance-reduced Jacobian-vector products. In Figure~\ref{fig:GMRESRes}
we show convergence plots for the GMRES solver for various numbers of
realizations and
system sizes. In our computations we choose $\eps = 10^{-5}$ for the
finite-difference approximation of the Jacobian, set model parameters as in E3
in Table~\ref{tab:params} and employ the in-built Matlab function \texttt{gmres} with
$\texttt{restart}=20$, $\texttt{tol}=10^{-5}$, $\texttt{maxit}=20$. As we can see,
the linear problems are well behaved. As expected, the linear iterations necessary to
obtain convergence decrease as we increase the number of realizations, but increase
with the system size. 

\subsection{Convergence of Newton-GMRES}
\begin{figure}
  \centering
  \includegraphics{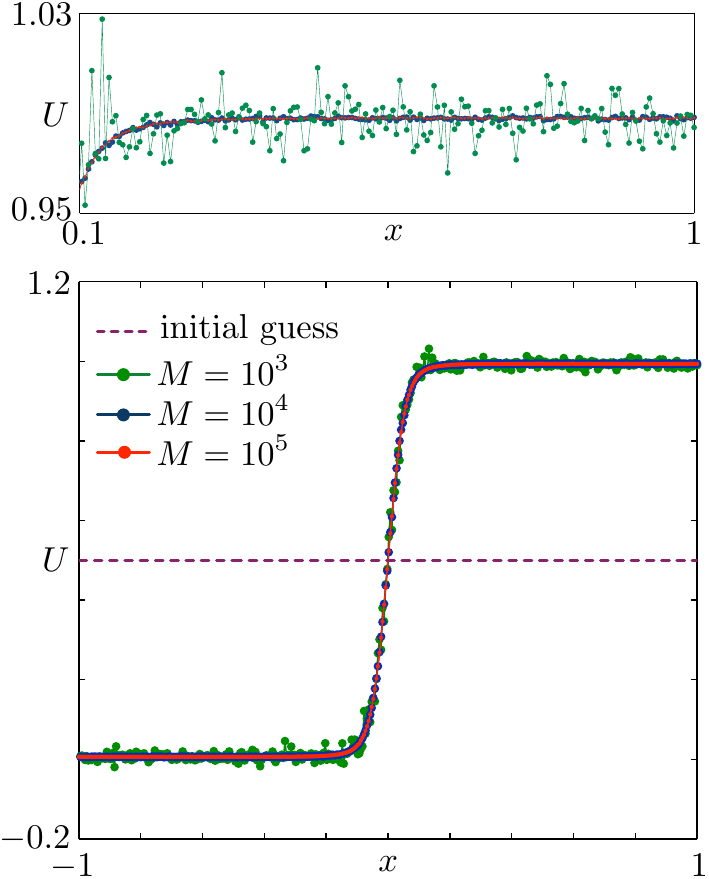}
  \caption{Macroscopic fronts computed with the Newton-GMRES solver (see also
  Figures~\ref{fig:NGMRESConv}(c) and~\ref{fig:NGMRESConv}(d)). The initial guess for
  all cases is a mixed state $U(x)\equiv 0.5$. The inset on top shows that the noise
  in the macroscopic profile is controlled by increasing the number of
  realizations.}
  \label{fig:NGMRESSol}
\end{figure}
\begin{figure}
  \centering
  \includegraphics{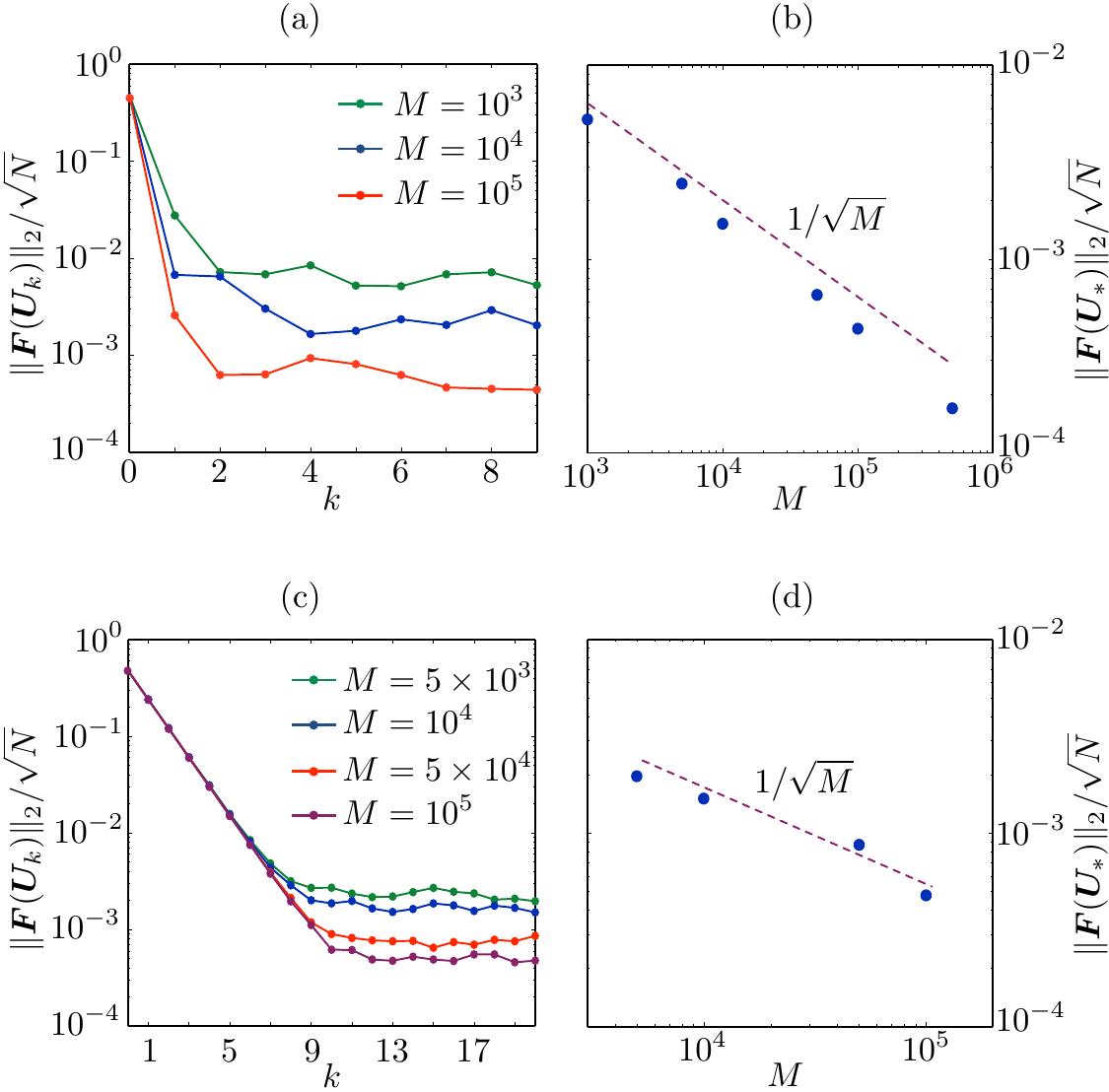}
  \caption{Convergence of the Newton-GMRES solver to compute a macroscopic front
  solution. The initial guess is a mixed state $U(x) \equiv 0.5$ and the Newton steps
  are damped with a constant factor $0.5$. Parameters as in E1
  of Table~\eqref{tab:params}, with the exception of $\alpha = 0.7$. Panel
  (a): convergence for various values of the number of realizations $M$, for $N=40$.
  Panel (b): the best achieved tolerance in the experiment of panel (a) is an
  $\mathcal{O}(1/\sqrt{M})$. Panels (c) and (d): the computations of panels (a)
  and (b) are repeated with $N=400$. Solution profiles for panels (c) and (d)
  are shown in Figure~\ref{fig:NGMRESSol}.}
  \label{fig:NGMRESConv}
\end{figure}
We now proceed to the computation of a single front, for fixed values of the control
parameters, using Newton-GMRES method. We start our iterations
with a mixed state, $U_n = 0.5$ for all $n$, and converge to a front, whose profile for
various values of the number of realizations is shown in Figure~\ref{fig:NGMRESSol}.
In panels $a$ and $c$ of Figure~\ref{fig:NGMRESConv} we show convergence plots of the
Newton-GMRES solver for $N=40$ and $N=400$, with various numbers of realizations. 
In these plots we scale the residual by $\sqrt{N}$, so as to compare performances
with varying system sizes. The Newton steps are built around the linear solves
described in Section~\ref{subsec:GMRESConvergence} and each Newton update is damped
by a constant factor $0.5$.

In the low-dimensional case, $N=40$, the solver achieves convergence in less
than $4$ iterations and then residuals plateau and begin to oscillate, as expected
(panel a of Figure~\ref{fig:NGMRESConv}). The onset of these oscillations is an
indication of the best tolerance that we can achieve with the nonlinear solver for a
fixed number of realizations: such tolerance is of $\mathcal{O}(1/\sqrt{M})$, as is
shown in panel b. In the high-dimensional case, $N=400$, a similar scenario occurs,
albeit more iterations are needed to achieve convergence. 
We point out that the experiment of Figure~\ref{fig:NGMRESConv} represents a severe
test for the nonlinear solver, in that we have chosen a poor initial guess (we start
from a mixed state to obtain a front). During continuation, initial guesses are
provided by a tangent or secant predictor step, resulting in much faster
convergence. We also remark that the convergence of the nonlinear solver is linear,
as expected, since we are using damped Newton updates.

\section{Bifurcation study of vendor lock-in model}
\label{sec:bifurcationResults}
In this section we present the results of coarse-grained numerical bifurcation
analyses of the lock-in model. The bifurcation diagrams have been computed with a
simple natural continuation method, that is, we start from a known solution
to the steady-state problem~\eqref{eq:ef-principle-newton}, increment the
continuation parameter and solve a new problem using the previous solution as an initial
guess. Even though this is not an optimal continuation strategy (as it does not allow
one to go past folds with a single run), we employ it here mainly for its simplicity,
keeping in mind that pseudo-arclength continuation with tangent or secant prediction
steps can easily be implemented.

\subsection{Continuation of homogeneous steady states}\label{sec:bif-homo}
We compute branches of homogeneous states using weighted lifting operators for a
population of $N=400$ agents. Since $N$ is large, we can compare our results with
branches of fixed points of the approximate evolution map~\eqref{eq:coarse1DMap}. In
Figure~\ref{fig:pitchforkHomogeneousComparison} we compare a few branches noting that
discrepancies are due to the finite size of the system. We also point out that each
point on the branch is the solution of a $400$-dimensional coarse system: in
principle we could have tracked the solution of a simple one-dimensional coarse system,
since the solutions we are finding are homogeneous; however, this experiment provides
a benchmark for our method and prepares us for the 
continuation of fronts.

\begin{figure}
  \centering
  \includegraphics{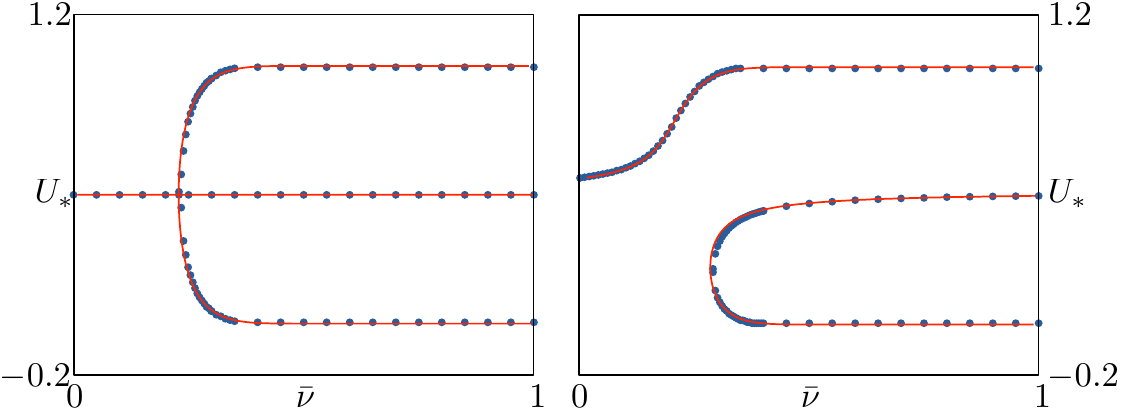}
  \caption{Coarse bifurcation diagram for homogeneous states. With blue dots we
  indicate the result of the equation-free continuation using weighted lifting
  operators. Red lines correspond to branches of fixed points of the approximate
  analytic evolution map~\eqref{eq:coarse1DMap}, as in
  Figure~\ref{fig:analyticContinuation}.
  The bifurcation parameter is the average $\anu$ of quality perception $q_n$, as
  reported in Equation~\eqref{eq:parMoments}. Other parameters: $\amu=0$ (left) and
  $\amu=0.04$ (right), $\Delta \mu=0$, $\alpha=0$, $\axi=0.236$, $\azeta=0$, $\beta = 10^8$, $N=400$.
  For the equation-free computations we use $M=10^4$ realizations and solve a
  $400$-dimensional coarse system.}
  \label{fig:pitchforkHomogeneousComparison}
\end{figure}

\begin{figure}
  \centering
  \includegraphics{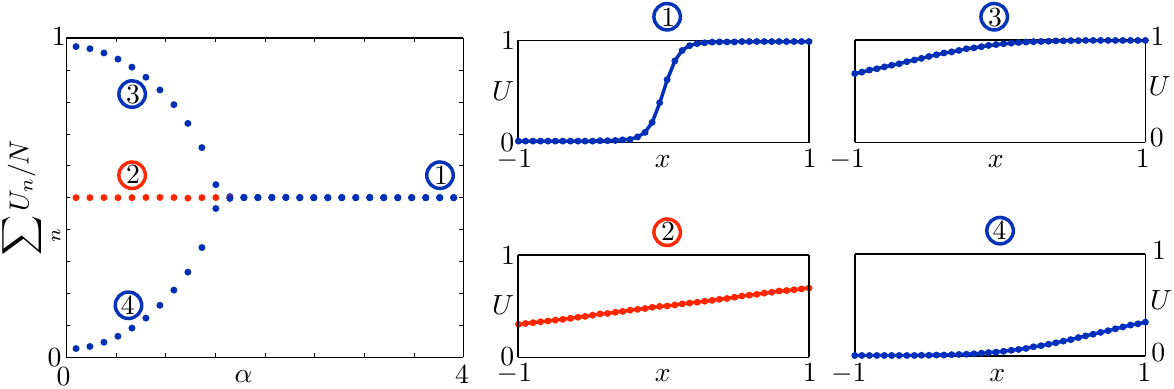}
  \caption{Coarse bifurcation diagram for fronts. The bifurcation parameter is the
  slope $\alpha$ of the profile of the average quality perception $q_n$, as reported
  in Equation~\eqref{eq:parMoments}. For large $\alpha$ a stable macroscopic front is
  formed (solution 1). As $\alpha$ is decreased, the front destabilises at a
  symmetry-breaking instability, generating partially locked-in states (solutions 3
  and 4). A fixed continuation step of $\Delta \alpha = 0.14$ has been used in the
  computations; other parameters as in E3 of Table~\ref{tab:params}.}
  \label{fig:pitchforkFrontCoarse}
\end{figure}

\subsection{Continuation of fronts as a function of $\alpha$}
In Figure~\ref{fig:pitchforkFrontCoarse} we show a coarse bifurcation diagram of
macroscopic fronts for a one-dimensional lattice with $N=40$ agents. We recall here
that fronts are observed in the inhomogeneous lock-in model, for which 
\[
q_n \sim \mathcal{N}(\bar \mu(x_n), \bar \xi), \qquad \mu(x_n) = \bar \mu + \Delta \mu
\tanh( \alpha x_n ).
\]
We choose parameters as in E3 of Table~\ref{tab:params}, with the exception of the
slope of the sigmoid, $\alpha$, which is the continuation parameter. The computations
are performed with $5\times 10^4$ realizations, using a linear tolerance of
$10^{-3}$, variance-reduced Jacobian-vector products with $\eps = 10^{-5}$, a relative
nonlinear tolerance of $2 \times 10^{-3}$ and continuation steps $\Delta \alpha =
0.14$.
As we decrease $\alpha$, the stable front (labelled $1$) loses stability at a
symmetry-breaking bifurcation, giving rise to two partially locked-in states
(labelled $3$ and $4$). These solutions correspond to the ones found via direct
numerical simulations (see coarse profiles in Figure~\ref{fig:inhomStates}). 
As we increase $\alpha$, the stable front (labelled $1$) becomes steeper: owing to
our particular choice of $\mu(x)$, the limit of large $\alpha$
corresponds to two competing radical factions of the same size, and so in this limit
the distribution of the average choice approaches a step function.

\begin{figure}
  \centering
  \includegraphics{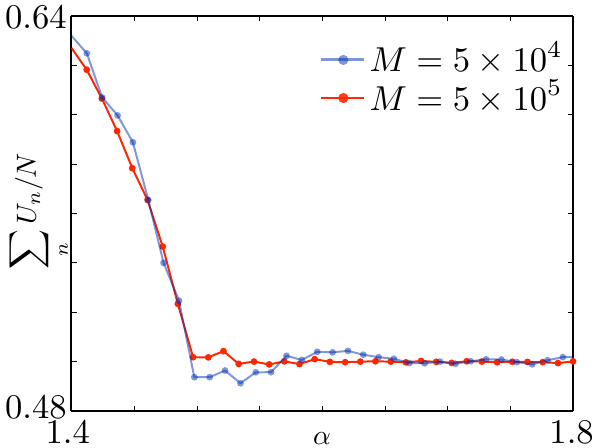}
  \caption{Effect of the number of realizations $M$ on the numerical
  continuation step size. We compute one of the branches in
  Figure~\ref{fig:pitchforkFrontCoarse} in the proximity of the
  symmetry-breaking bifurcation, with a much smaller continuation step, $\Delta
  \alpha = 0.012$. The continuation with $5 \times 10^{4}$ realizations and nonlinear
  tolerance $2 \times 10^{-3}$ (blue) is now affected by noise, which is reduced by
  setting $M = 5 \times 10^5$ and a tolerance of $7 \times 10^{-4}$ (red).}
  \label{fig:pitchforkFrontZoomIn}
\end{figure}

As expected, the number of realizations influences the continuation step size: in
Figure~\ref{fig:pitchforkFrontZoomIn} we re-compute one of the branches in
Figure~\ref{fig:pitchforkFrontCoarse} in the proximity of the
symmetry-breaking bifurcation, with a much smaller continuation step, $\Delta
\alpha = 0.012$. The continuation with $M = 5 \times 10^{4}$ and a relative
nonlinear tolerance of $2 \times 10^{-3}$ (blue curve) is now affected by noise,
which can be reduced by increasing the number of realizations to $5 \times 10^5$ and
set a tolerance of $7 \times 10^{-4}$ (red curve).

During continuation, we infer stability of a coarse solution $\vect{U}_\ast$ by
computing eigenvalues of $\vect{D}\vect{F}(\vect{U}_\ast) = \vect{I} -
\vect{D}\vect{\Phi}^M_T(\vect{U}_\ast)$. Since both $N$ and $M$ are relatively small,
we form $\vect{D}\vect{F}(\vect{U}_\ast)$ using the finite difference
approximation~\eqref{eq:ef-principle-jacFD} $N$ times and then compute the
full spectrum at once. For larger system sizes, matrix-free Arnoldi iterations can be
employed to compute only the leading eigenvalues. In Figure~\ref{fig:frontEigs}
we plot the most unstable eigenvalue as a function of the bifurcation parameter,
showing that the symmetry-breaking instability occurs at $\alpha \approx 1.5$.
As a further remark on the accuracy of the variance-reduced Jacobian
calculations, we plot the full spectrum for selected values of the continuation
parameter, showing a clear separation between the leading real eigenvalue and a
tight cluster of eigenvalues at the origin.

\begin{figure}
  \centering
  \includegraphics{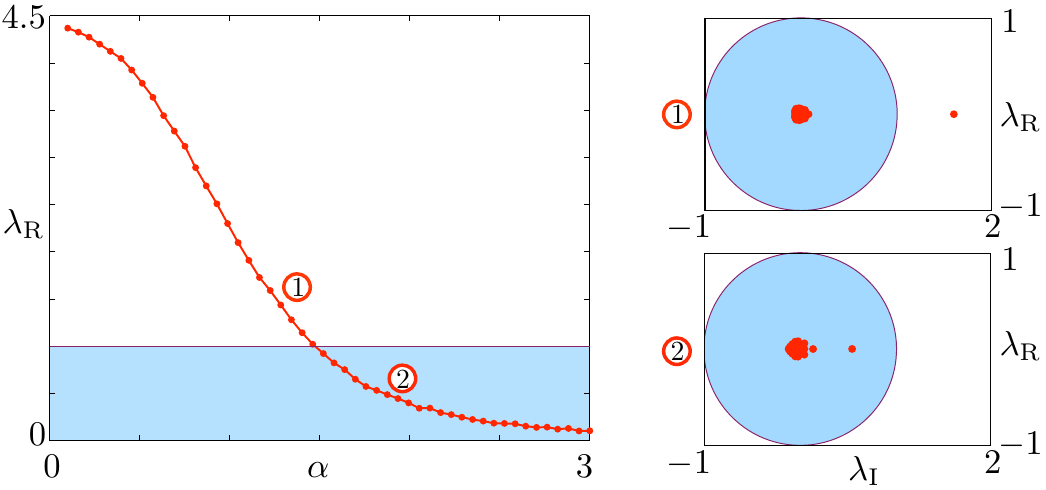}
  \caption{Eigenvalues of macroscopic front solutions on the symmetric branch of
  the bifurcation diagram in Figure~\ref{fig:pitchforkFrontCoarse}. Left: the real
  part of the most unstable eigenvalue is plotted as a function of the
  continuation parameter $\alpha$. Right: representative spectra along the
  branch.}
  \label{fig:frontEigs}
\end{figure}

\section{Conclusion}\label{sec:conclusions}
In this paper we have computed and continued in parameter space coarse-grained
states for an ABM of consumer lock-in with heterogeneous agents. We first
considered the simple case of homogeneous agents and found an explicit macroscopic
evolution map for the expectation of the 
mean purchase. 
As expected, this coarse
description leads to a scenario similar to the one found for linearly-coupled
oscillators subject to a double-well potential~\cite{Barkley2006y}: the first moment
map stabilises metastable locked-in states, which arise at a pitchfork bifurcation of
the coarse map; introducing a homogeneous preference for one of the two products has
the effect of breaking the pitchfork or, in the heat bath analogy, to introduce an
asymmetric double-well potential.

The more interesting and challenging case of agents split into factions with opposite
preferences leads to the formation of 
monotonically-increasing
macroscopic patterns, which have been computed using a 
large-dimensional
coarse
description. Our analysis reveals that, as the inhomogeneity becomes less pronounced,
fronts undergo a symmetry-breaking instability. The resulting stable patterns are not
fully locked-in, in that they feature pockets of resistance of each faction. An
interesting future extension of this model could include a more granular
modulation in the agents' preferences with nearest neighbour
coupling: in this case the lock-in model would be posed on a spatially-extended
lattice and oscillatory one-dimensional inhomogeneities or two-dimensional randomly
distributed factions 
can be studied with the method presented
here. In particular, the formulation of the two-dimensional coarse
problem would follow the same steps as the one-dimensional case, even though it
would naturally lead to a larger coarse system size. A further
extension could include agent motility: in this way it would be possible to
study how coherent spatio-temporal states, such as travelling fronts or bumps, are
related to the spatial heterogeneity in agent preferences.

The core result of the paper is a strategy to evaluate variance-reduced
Jacobian-vector products in equation-free methods. The main idea behind our approach
is to exploit the non-uniqueness of the lifting operator to obtain a coarse
time-stepper which depends smoothly on the coarse variables. 
In practice, this is achieved by using weighted averages in the restriction step and
pre-computing weights during the lifting step. We have shown that a direct
consequence of using weights is that we gain full control over the linear solves,
leading to well-behaved GMRES iterations and, ultimately, to nonlinear convergence
for large-dimensional coarse descriptions.

In order to assess the efficiency of the weights, we draw a comparison between
weighted and unweighted Newton steps when the number of realizations $M$ is fixed.
In the unweighted case, each Newton step requires $1$ evaluation of $\Phi^T_M$ (1
Bernoulli sampling, $M$ evolutions, 1 average) and then, for each GMRES step, a
further evaluation of $\Phi^T_M$ involving $M$ further evolutions. In the
weighted case, each Newton step requires $1$ weighted evaluation of $\Phi^T_M$
($1$ Bernoulli sampling, $M$ evolutions, $1$ manipulation of the constraint
matrix, $1$ linear solve, $1$ weighted average) and then, for each GMRES step,
$1$ linear solve and $1$ weighted average.  Considering the improved GMRES and
Newton-GMRES convergence, weighted operators seem more efficient, especially
when running evolution steps is expensive. 

We remark that, for the case under consideration, it was not possible to make a
quantitative comparison of the efficiencies of weighted and unweighted coarse
time-steppers, since the unweighted Newton-GMRES solver failed to converge for the
inhomogeneous case. This reinforces the idea that, in large-dimensional coarse
systems, noise can be harmful and variance-reduced Jacobian evaluations become an
important ingredient in equation-free methods.
Furthermore, weighted operators could be employed also in smaller coarse systems, such
as the ones deriving from Galerkin discretizations of spatially-extended
systems~\cite{Gear2002a} or from chemical systems of moderate
sizes~\cite{Hoyle2012a}.  

A natural question arises as to whether weighted operators are applicable to other
types of coarse-grained models. In Section~\ref{sec:eq-free} we have presented
weighted operators for the lock-in model, for which microscopic variables are binary
numbers, but we envisage that similar ideas will be relevant in models where the
microscopic variables are real numbers. In particular it seems plausible to assume
that the minimization problem~\eqref{eq:minProb}--\eqref{eq:constr3} will remain
valid if $\vect{u}^m \in \mathbb{R}^N$. 
Our current choice of the optimisation problem for weights was driven by the following
criteria: 
\begin{enumerate}
\item The weighted realisations should satisfy the restriction exactly;
\item The weights should converge to $1$ as $M\to\infty$;
\item The weights should depend continuously on the macroscopic state;
\item The weights should introduce minimal perturbations to the sampled probability distributions.
\end{enumerate}
Other procedures to determine the weights are conceivable: we could allow only a
limited number of weighting factors, or use a different norm in the target function.
When only allowing a limited number of weights, one clearly imposes additional
artefacts on the represented probability distribution of realisations. Neither the
resulting artefacts nor the effect of the choice of norm in the target function have
been systematically studied in this work.
These aspects, together with a more rigorous justification of weighted operators,
will be the subject of future work.

\section*{Acknowledgements}
DA acknowledges the University of Nottingham Research Development Fund,
supported by the Engineering and Physical Sciences Research Council (EPSRC). DA
and RH acknowledge partial funding from the EPSRC grant EP/H021779/1. 
The work of GS was partially supported by the Research Council of KU
Leuven through grant OT/13/66, by the Interuniversity Attraction Poles Programme
of the Belgian Science Policy Office under grant IUAP/V/22, and by the Research Foundation -- Flanders (FWO) through grant G.A003.13. We are grateful to 
Andrew Archer, Nigel Gilbert, David Lloyd, Alastair Rucklidge, Jan Sieber and
Anne Skeldon for stimulating discussions about this work.

\bibliographystyle{plain}
\bibliography{lockin}

\end{document}